%% file: final_main.tex
\documentclass[journal]{IEEEtran}
\usepackage{cite}
\usepackage{amsmath,amssymb,amsfonts}
\usepackage{graphicx}
\usepackage{textcomp}

\usepackage{todonotes}
\usepackage{xcolor}
\usepackage{dsfont}
\usepackage{cite}
\usepackage{graphicx}
\usepackage{mathrsfs}
\usepackage{amsbsy}
\usepackage{comment}
\usepackage{hyperref}
\usepackage{siunitx}
\usepackage{soul}
\usepackage{algpseudocode}
\usepackage{algorithm}
\usepackage{caption}
\usepackage{subfigure}
\usepackage{graphics,framed, wrapfig}
\usepackage{upgreek}
\usepackage{bm}
\usepackage{amsthm}

\newcommand{\MFa}[1]{\textcolor{blue}{#1}}

\definecolor{MFabg}{RGB}{76, 191, 229}

\def \nI {n_{I}}
\def \nM {n_{M}}

\input{shortcuts_math}
\ifCLASSINFOpdf
\else
\fi
\begin{document}
%
\title{Secure Consensus via Objective Coding:\\  Robustness Analysis to Channel Tampering}
%
%
%

\author{Marco Fabris and Daniel~Zelazo,~\IEEEmembership{Senior Member,~IEEE}
\thanks{Manuscript received January 24, 2022; revised April 7, 2022; accepted May 19, 2022. This work was supported in part by the United States-Israel Binational Science Foundation (BSF) under Grant 2017658, and in part by the United States National Science Foundation (NSF) under Grant 1809315. This article was recommended by Associate Editor J. Catalão. (Corresponding author: Marco Fabris.)}
\thanks{The authors are with the Faculty of Aerospace Engineering,
	Technion—Israel Institute of Technology, Haifa 3200003, Israel (e-mail:
	marco.fabris@campus.technion.ac.il; dzelazo@technion.ac.il).
	Color versions of one or more figures in this article are available at
	https://doi.org/10.1109/TSMC.2022.3177756.}
\thanks{Digital Object Identifier 10.1109/TSMC.2022.3177756}
}

%
%

\markboth{Journal of \LaTeX\ Class Files,~Vol.~14, No.~8, August~2015}%
{Shell \MakeLowercase{\textit{et al.}}: Bare Demo of IEEEtran.cls for IEEE Journals}
%



\maketitle

\begin{abstract}
This work mainly addresses continuous-time multiagent consensus networks where an adverse attacker affects the convergence performances of said protocol. In particular, we develop a novel secure-by-design approach in which the presence of a network manager monitors the system and broadcasts encrypted tasks (i.e., hidden edge weight assignments) to the agents involved. Each agent is then expected to decode the received codeword containing data on the task through appropriate decoding functions by leveraging advanced security principles, such as objective coding and information localization. Within this framework, a stability analysis is conducted for showing the robustness to channel tampering in the scenario where part of the codeword corresponding to a single link in the system is corrupted. A tradeoff between objective coding capability and network robustness is also pointed out. To support these novelties, an application example on decentralized estimation is provided. Moreover, an investigation of the robust agreement is as well extended in the discrete-time domain. 
Further numerical simulations are given to validate the theoretical results in both the time domains.
\end{abstract}


\IEEEpeerreviewmaketitle

\section{INTRODUCTION}
\label{sec:intro}

The consensus protocol has become a canonical model for the study of \textit{multiagent systems} (MASs),
groups of autonomous entities (\textit{agents}) that interact with each other to solve problems that are beyond the capabilities of a single agent \cite{Mahmoud2020}.
Such architectures are characterized by a cooperative nature
that is robust and scalable. Robustness refers to the ability of a system to tolerate the failure of one or more agents, while scalability originates from system modularity. 
Because of these advantages, networked architectures based on MASs have become popular in several cutting-edge research areas such as 
the
Internet-of-Things \cite{SavaglioGanzhaPaprzycki2020}
and Cyber-Physical Systems \cite{KouicemRaievskyOccello2020}.
As stated in \cite{OlfatiSaberFaxMurray2007},
within such networks of agents, ``consensus" means to reach an agreement w.r.t. a certain quantity of interest that depends on the state of all agents. A ``consensus algorithm" (or agreement protocol) is an interaction rule that specifies the information exchange between an agent and all of its neighbors in the network such that agreement is attained. 

Recently, the increasing demand for safety and security measures in the most advanced technologies have skyrocketed in many fields, including that of MASs  \cite{WolfSerpanos2018,ZhihanDongliangRanarn2021}.
In fact, the concerns about protection of networked systems from cyber-physical attacks are not new, and have attracted a fair amount of attention in the engineering community. 
As a consequence, several approaches to improve the security of such systems or understand their vulnerabilities have been developed \cite{SchulzeDarupAlexandruQuevedo2021}.
A first step in this direction is to analyze the robustness properties of consensus networks. Few examples of different connotations addressing this desired property are given by one or a combination of the following requirements: (a) the network reaches an $\varepsilon$-consensus, i.e., for all $(i,j) \in \Emc$ it holds $\lim_{t \rightarrow \infty} \left\|x_{i}-x_{j}\right\|_{2} \leq \varepsilon$, for some $
\epsilon>0$ \cite{DibajiIshii2015}; (b) a subset of the network vertices converges to an agreement \cite{MustafaModaresMoghadam2020}; (c) a cost function of the state that serves as a performance index for the level of agreement is expected to decrease or stay below a certain given threshold \cite{YanAntsaklisGupta2017}; (d) the network fulfills consensus in spite of the presence of ``small"-magnitude perturbations altering the agent dynamics \cite{ZelazoBurger2017}.

\textit{Related works}: In the literature, many techniques for secure consensus or synchronization within a network are available. Most of them rely on the concept of resilience, ensuring \emph{robustness} to attacks or faulty behaviors. 
In \cite{WeerakkodyLiuSon2017}, classic tools from system theory are applied on networks modeled as discrete-time MASs in order to design observers and algebraic tests with the goal of identifying the presence of misbehaving agents. 
These identification-based techniques require a deep understanding of the processes to be controlled and thus their design is quite complex. Also, to the best of our knowledge, continuous-time MASs have not been studied by means of those tools yet. 
In \cite{DibajiIshii2015,WuXuNing2020} 
part of the information being exchanged by the neighbors to a certain agent is chosen and then fully neglected via thresholding mechanisms. These selections are executed according to a given order that imposes some priority on the information itself to achieve attack mitigation. Such an approach can however lead to strong biases, since it is possible that the designated order is not adequate. Moreover, global information on the network topology is required in the design leading to a centralized implementation  (see also \cite{LeBlancHeathZhang2013}). 
In \cite{TrentelmanTakabaMonshizadeh2013}, robust synchronization is attained through protocols based on regulators that make use of a state observer. These methods require the computation of maximal real symmetric solutions of certain algebraic Riccati equations, also involving weighting factors that depend on the spectral properties of the network graph. There have been additional works focusing on resilient architectures for microgrids \cite{YassaieHallajiyanSharifi2021}, and MASs under denial-of-service attacks \cite{DuWangDong2021, ZuoCaoWang2021, FengXieWang2021}. Lastly, a thriving part of this area directs its effort toward investigations coping with ``privacy preserving consensus'' \cite{GaoWangHe2021,RuanGaoWang2019,MoMurray2017,WangPrivacyPreserving2019,AltafiniPrivacyConsensus2019}. 
However, in contrast to this study, the attention has been focused much more on discrete-time systems 
or concealing the information being exchanged by nodes, in order to preserve privacy or relevant data, such as initial conditions of the network states.


\textit{Adopted framework:} Notwithstanding the meaningful novelties, many of these works lack a simple, scalable, flexible and distributed principle that renders a consensus MAS resilient to specific cyber-physical threats that \textit{aim at slowing down the convergence or destabilizing the network by attacking its links}. This approach thus seeks to preserve confidentiality, integrity and availability in the system itself starting by the design of resilient network connections. Instead of developing tools to secure existing systems, we provide inherently secure embedded measures that guarantee robust consensus convergence. 

\textit{Methodology:} Our approach is not meant to replace usual security measures; conversely, it furnishes further innovative security mechanisms based on the \textit{secure-by-design philosophy}, popular in software engineering \cite{Kreitz2019}. 
The core of this study consists in the development of a secure-by-design approach and its application to the consensus theory. To this aim, we take the point of view of a {\it network manager} pitted against an {\it attacker}. 
The goal of the network manager is to supply a networked system with an objective to be achieved. 
The goal of the attacker is to disrupt the operation of the system and prevent it from reaching its goal. 
Generally, such sensitive information may lay in the state of the agents, or be the global objective of the system. 
Our proposed solution approach is built upon three overarching principles: (i) embed the agents with hidden security measures, (ii) control the information given to the agents and (iii) make the dynamics robust and resilient.
The first principle arises from the fact that a certain amount of freedom is often available in the design stage. One can, for instance, adopt encryption methods to conceal the objective the network is aiming at, namely \textit{objective coding} can be leveraged as a security measure whenever an attacker is attempting to inject a malicious signal in the system. To this purpose, encoding/decoding functions are employed to serve as an encryption mechanism in order to keep hidden the real network objective.
The second principle stems from the fact that a MAS is designed, in general, to fulfill a certain situation-specific task. Thus, the information spread among agents needs to be quantified and maintained to the strict minimum, leading to the study of \textit{information localization}.
Finally, the last principle strives to render the dynamics as \textit{robust} as possible to attacks, while ensuring that the objective can be reached with limited information. 

\textit{Contributions}: The contributions of this work are threefold. 
\begin{itemize}
	\item[1.] A \textit{secure-by-design} consensus protocol is devised to satisfy principles (i)-(iii) within a given multiagent network under attack. The tradeoff between information encryption and robust convergence is analyzed.
	\item[2.] A stability and robustness analysis is performed both in continuous and discrete time to show that the proposed protocol is resilient to small perturbations affecting  the reception of encrypted edge weights.
	\item[3.] An application to decentralized estimation involving the decentralized power iteration algorithm is presented to highlight the validity of our approach.
\end{itemize}

\textit{Paper outline:} The remainder of the paper is organized as follows.
Sec. \ref{sec:prelim_models} introduces the preliminary notions and models for multiagent consensus.
In Sec. \ref{sec:sbdconsensus}, our proposed strategy to secure the design of consensus is developed and discussed. Sec. \ref{sec:robust-analysis} provides its robustness analysis when the network is subject to channel tampering modeled as single-edge-weight perturbation, while Sec. \ref{sec:PI-ACE} reports on an application to decentralized estimation. Sec. \ref{sec:dt_ext} extends this study in the discrete-time domain. 
Numerical simulations assessing the obtained theoretical results are reported in Sec. \ref{sec:simulations} and conclusions are sketched in Sec. \ref{sec:conclusions}.

\textit{Notation:} The set of real, real non-negative, and complex numbers are denoted with $\Rset{}$, $\Rset{}_{\geq 0}$, and $\mathbb{C}$, respectively, while $\Re[\varsigma]$ and $\Im[\varsigma]$ indicate the real and imaginary parts of $\varsigma \in \mathbb{C} $. Symbols $\onesvec{l} \in \Rset{l}$ and $\zerovec{l} \in \Rset{l}$ identify the $l$-dimensional (column) vectors whose entries are all ones and all zeros, respectively, while $\eye{l} \in \Rset{l\times l}$ and $\zerovec{l \times l} \in \Rset{l\times l}$ represent the identity and null matrices, respectively. 
We indicate with $\mathbf{e}_{l}$ the canonical vector having $1$ at its $l$-th component and $0$ at all the others.
The Kronecker product is denoted with $\kronecker$. Let $\Omega \in \Rset{l \times l}$ be a square matrix. Relation $\Omega \succeq 0$ means that $\Omega$ is symmetric and positive semi-definite. The notation $[\Omega]_{ij}$ identifies the entry of matrix $\Omega$ in row $i$ and column $j$, while $\left\|\Omega\right\|$, $\Omega^{\top}$, and $\Omega^\dagger$ indicate its spectral norm, its transpose, and its Moore-Penrose pseudo-inverse. Operators $\ker(\Omega)$, $\mathrm{col}_{l}[\Omega]$, and $\mathrm{row}_{l}[\Omega]$ indicate each the null space, the $l$-th column, and the $l$-th row of $\Omega$. The $i$-th eigenvalue of $\Omega$ is denoted by $\lambda_{i}^{\Omega}$. The space spanned by a vector $\omega \in \Rset{l}$, with $i$-th component $[\omega]_{i}$, is identified by $\angb{\omega}$. The Euclidean and infinity norms of $\omega$ are denoted with $\left\|\omega\right\|_{2}$ and $\left\|\omega\right\|_{\infty}$. 
Finally, $\Bomega=\vec{i=1}{l}{\omega_{i}}$ defines the vectorization operator stacking vectors $\omega_{i}$, $i=1,\dots,l$ as $\Bomega = \begin{bmatrix}
	\omega_{1}^{\top} & \dots & \omega_{l}^{\top}
\end{bmatrix}^{\top}$; whereas, $\diag{i=1}{l}{\varsigma_i}$ is a diagonal matrix with $\varsigma_i \in \mathbb{R}$, $i=1,\dots,l$, on the diagonal.

\section{Preliminaries and models}
\label{sec:prelim_models}

In this section, preliminary notions and models for MASs are introduced along with a brief overview on consensus theory and robustness in consensus networks. 


An $n$-agent system can be modeled through a weighted graph $\mathcal{G}=\left(\mathcal{V},\mathcal{E},\Wmc\right)$ so that each element in the \textit{vertex set} $\mathcal{V}=\left\{1, \dots, n\right\}$ is related to an agent in the group, while the \textit{edge set} $\mathcal{E}\subseteq \mathcal{V}\times \mathcal{V}$ characterizes the agents' interactions in terms of both sensing and communication capabilities. Also, $\Wmc = \{w_{k}\}_{k=1}^{m}$, with $m = |\Emc|$, represents the set of weights assigned to each edge. Throughout the paper, bidirectional interactions among agents are supposed, hence $\mathcal{G}$ is assumed to be \textit{undirected}. 
The set $\mathcal{N}_i=\left\{j\in\mathcal{V}\setminus\{i\} \;|\;(i,j)\in \Emc \right\}$ identifies the \textit{neighborhood} of the vertex $i$, i.e., the set of agents interacting with the $i$-th one and the cardinality $d_{i} = |\mathcal{N}_{i}|$ of neighborhood $\Nmc_i$ defines the degree of node $i$. 
Furthermore, we denote the \textit{incidence matrix} as $E \in \Rset{n\times m}$, in which each column $k \in \{1,\ldots,m\}$ is defined through the $k$-th (ordered) edge $(i,j) \in \Emc$, where $i<j$ is adopted w.l.o.g., and for edge $k$ corresponding to $(i,j)$ one has $[E]_{lk} = -1$, if $l = i$; $[E]_{lk} = 1$, if $l = j$; $[E]_{lk} = 0$, otherwise.
For all $k = 1,\ldots,m$, the weight $w_{k} = w_{ij} = w_{ji} \in \Rset{}$ is associated to $k$-th edge $(i,j)$, and $W=\diag{k=1}{m}{w_k}$ is the diagonal matrix of edge weights. 
Also, the \textit{Laplacian matrix} containing the topological information about $\graph$ is addressed as $L(\Gmc) = E W E^{\top} $ (see \cite{MesbahiEgerstedt2010}). Henceforward, we also assume that graph $\graph$ is \textit{connected} and $L(\Gmc) \succeq 0$, having eigenvalues $\lambda_{i}^{L}$, for $i = 1,\ldots,n$, such that $0 = \lambda_{1}^{L} < \lambda_{2}^{L} \leq \cdots \leq \lambda_{n}^{L}$. A sufficient condition to satisfy the latter requirement, which is adopted throughout the paper, is setting $w_{ij}> 0$ for all $(i,j)$.
Lastly, we let $w_{i} = \sum_{j \in \Nmc_{i}} w_{ij}$ and $\Psi_{\Gmc} = \max_{i=1,\ldots,n} w_{i}$ be the weighted degree of the $i$-th node and the maximum weighted degree of $\Gmc$, respectively.

We now provide an overview of the weighted consensus problem in MASs. Let us consider a group of $n$ homogeneous agents, modeled by a weighted and connected graph $\graph$. 
Let us also assign a continuous-time state $x_{i} = x_{i}(t) \in \Rset{D}$ to the $i$-th agent, for $i = 1,\dots,n$. The full state of the whole network can be thus expressed by $\x = \vec{i=1}{n}{x_i} 
\in X \subseteq \Rset{N}$, with $N=nD$. Consequently, the weighted consensus within a MAS can be characterized as follows.

\begin{defn}[Weighted Consensus\textcolor{white}{,}\cite{MesbahiEgerstedt2010}]\label{def:consensus}
	{An $n$-agent network achieves \emph{consensus} if $\lim_{t\rightarrow+\infty} \x(t) \in \Amc $, where $\Amc = \angb{\ones_{n}} \otimes \omega$, for some $\omega \in \mathbb{R}^{D}$, is called the \emph{agreement set}.
	}
\end{defn}

For a connected graph $\Gmc$ with positive weights, it is well known that the \textit{linear weighted consensus protocol}, given by
\begin{equation}\label{eq:LAP}
	\dot{\x} = - \L(\Gmc) \x,
\end{equation}
where $\L(\Gmc) = (L(\Gmc) \kronecker I_{D})$, drives the ensemble state to the agreement set \cite{MesbahiEgerstedt2010}.

We now review a robustness result for the consensus protocol with small-magnitude perturbations on the edge weights \cite{ZelazoBurger2017}.	In this setting we consider the perturbed Laplacian matrix $L(\Gmc_{\Delta^W}) = E (W+\Delta^W) E^{\top}$ for a structured norm-bounded perturbation $\Delta^W \in \bm{\Delta}^{W} = \{\Delta^W  \, : \, \Delta^W = \diag{k=1}{m}{\delta^{w}_k}$, $ \|\Delta^W \| \leq \bar{\delta}^W  \}$. When the injection attack is focused on a single edge, the following result (trivially extended from the corresponding one-dimensional case) is obtained relating the stability margin of an uncertain consensus network to the \textit{effective resistance} of an analogous resistive network \cite{KleinRandic1993}.

\begin{lem}[\cite{ZelazoBurger2017}] \label{thm:effective_resistance}
	Consider the nominal weighted consensus protocol \eqref{eq:LAP}. Then, for a single edge attack $\Delta^W = \delta_{uv}^{w} \mathbf{e}_z \mathbf{e}_z^{\top} \in \bm{\Delta}^{W}$ on the edge $z=(u,v)\in\mathcal{E}$, such that $ \delta_{uv}^{w}$ is a scalar function of $t$, the perturbed consensus protocol 
	\begin{equation}\label{eq:uncertainconsnet}
		\dot{\x} = -(L(\Gmc_{\Delta^W}) \kronecker I_{D}) \x
	\end{equation}
	is stable for all $\delta_{uv}^{w} $ satisfying 
	\begin{equation}\label{eq:fundamineqres}
		|\delta_{uv}^{w}| \leq \mathcal{R}_{uv}(\Gmc)^{-1},
	\end{equation}
	where $\mathcal{R}_{uv}(\Gmc) = [L^{\dagger}(\Gmc)]_{uu} - 2[L^{\dagger}(\Gmc)]_{uv} + [L^{\dagger}(\Gmc)]_{vv}$
	is the \textit{effective resistance} between nodes $u$ and $v$ in $\Gmc$.
\end{lem}
The result in \ref{eq:fundamineqres} is sharp in the sense it provides an exact upper bound on the robust stability of the system.  For multiple edge perturbations, a more conservative result based on the small-gain theorem is also provided \cite[Theorem V.2]{ZelazoBurger2017}.

\section{The secure-by-design consensus protocol}\label{sec:sbdconsensus}



In this work, we consider MASs which are led by a so-called \textit{network manager} 
providing encrypted objectives or parameters to the ensemble.  The MAS is also subject to an attack by an external entity aiming to disrupt the operation of the network.  
%
In this setup, agents receive high-level instructions from the network manager that describe a \textit{task} the agents have to achieve. Within the consensus framework, a task may consist in the assignment of edge weights, albeit the concept of ``task" may be varied according to further generalizations (e.g. nonlinear consensus) or depending on a specific multiagent framework (e.g. formation control). 
In particular, our attention is directed towards \emph{edge weight encryption}, since these dictate the convergence rate of protocol \eqref{eq:LAP} to the agreement. It is worth mentioning that the latter performance indicator plays a key role in the functioning of certain applications, e.g. those involving decentralized estimation \cite{YangFreemanGordon2010}, 
or in certain theoretical fields, as the problems related to averaged controllability \cite{Zuazua2014}.
Another crucial aspect in this setup is that \textit{the network manager is not conceived to operate as a centralized controller}. Indeed, this
does not send control signals to each agents for the system to  achieve a ``global objective'',
but instead  sends only a few parameters describing the objective to be achieved by the agents. Hence, the presence of the external manager does not invalidate any distributed architectures. Moreover, the use of a network manager that broadcasts the encoded objective to all the nodes is justified by the fact that each element of the network must be somehow made aware of the network parameters for their information exchange to occur correctly: we aim at the secure design for such a preliminary task assignment. In this consensus groundwork, our approach is indeed fully supported by the fact that optimal weight assignment problems requiring prior computations are of extreme relevance in literature and give birth to well-known research branches, e.g. the study of fastest mixing Markov processes on graphs \cite{BoydDiaconisXiao2004,SunBoydXia2006}. 

The kind of scenarios we envision then consists of two steps: firstly, the network manager broadcasts only a few signals, in which an (or a sequence of) objective(s) is encoded, and secondly, each agent follows a pre-designed algorithm or control law -- the consensus protocol, in this precise context -- depending on these local objectives. 
To this aim, objective coding and information localization represent the primary tools to encrypt tasks and spread the exchanged information. In the next lines, we provide more details about these principles, casting them on the consensus framework. 


\subsection{Objective coding and information localization}\label{ssec:objcod_infloc}

A major innovation of our approach lies in the introduction of objective decoding functions.  Here we assume that tasks are described by an encoded parameter $\theta$ that we term the \textit{codeword}.  The space of all tasks is denoted as $\Theta$. Each agent in the network then decodes this objective using its \textit{objective decoding function}, defined as $p_i : \Theta \rightarrow \Pi_i $, where  
$\Pi_i$ depends on the specific application (e.g. $\Pi_i \subseteq \Rset{n}$ within the consensus setting). Functions $p_i$ represent a secure encryption-decryption mechanism for the information describing the task being received. 
For $\theta \in \Theta$, $p_i(\theta)$ is called the \emph{localized objective}. Whereas, if $\theta \notin \Theta$, $p_i(\theta)$ may not be calculable; however, any agent receiving such a codeword may launch an alert, since this can be seen as an attack detection.  A possible example of this framework is to have $\Theta$ be a Euclidean space (e.g. the identity function), and $p_i$ be a projection onto some of the the canonical axes in the Euclidean space.
In other words, the common case in which $p_i$ are projection functions (e.g., $p_i(\theta) = \theta_i \in \Theta \subseteq \Rset{n^{2}}$ when $\theta := \vec{i=1}{n}{\theta_i}$, $\theta_i \in \mathbb{R}^{n}$) justifies the abuse of language of calling $\theta$ the objective. 
Moreover, we assume that the codewords $\theta$ are transmitted as in a \textit{broadcast mode}, that is the network manager broadcasts the objective $\theta$ in an encoded manner. Each agent is equipped with an individually designed function $p_i$ which extracts from $\theta$ the relevant part of the objective. Most importantly, the encoding and decoding mechanisms are assumed unknown to the attacker.



In addition to objective coding, information localization, the process by which only parts of the global variables describing the system are revealed to the agents, is fundamental in this design approach. So, to conclude, we let  $h_i(\x):X\to Y_i$, with $Y_i\subseteq X$, represent the \textit{information localization about the state of the ensemble} (containing $n$ agents) for agent $i$.

\subsection{Secure-by-design consensus dynamics}\label{ssec:SBDCdyn}

With the above conventions, principles and architecture, the general description of agent~$i$ can be expressed by
\begin{align}\label{agent_dynamics}
	\dot x_i = f_i(\x, u_i(h_i(\x), p_i(\theta))), \quad i = 1,\ldots,n,
\end{align}
where $u_{i} = u_i(h_i(\x), p_i(\theta))$ is the control or policy of agent~$i$, which can only depend on the partial knowledge of the global state and objective coding.

Now, since in this paper we are coping with secure linear consensus protocols, dynamics in \eqref{agent_dynamics} is specified through the following characterization dictated by the nominal behavior in \eqref{eq:LAP}. Firstly, the objective coding is established through the non-constant functions $p_{i}: \Theta \rightarrow \Pi_i \subseteq \Rset{n}$, such that $ [p_{i}]_{j} := p_{ij} $,  with 
\begin{equation}\label{eq:pijbasicchar}
	p_{ij}(\theta)  = \begin{cases}
		w_{ij}, & \text{ if } (i,j) \in \Emc \\ 
		0, &  \text{ otherwise}.
	\end{cases}
\end{equation}
The values $w_{ij}$ in \eqref{eq:pijbasicchar} coincide with the nominal desired consensus weights set by the network manager.
Secondly, the information localization about the global state $\x$ is expressed by means of $h_{i}(\x) : X \rightarrow  Y_{i} \subseteq \Rset{D \times n} $, such that $\mathrm{col}_{j}[h_{i}(\x)] := h_{ij}(\x(t)) \in \Rset{D}$ with $h_{ij}(\x)=x_{i}-x_{j}$, if $(i,j) \in \Emc$; $h_{ij}(\x)=\zerovec{D}$, otherwise.
As a consequence, the peculiar dynamics $f_i(\x, u_i)$ for the $i$-th agent involved in the \textit{secure-by-design consensus} (SBDC) is determined by
\begin{align}\label{eq:secure_consensus} 
	f_i(\x, u_i(h_i(\x), p_i(\theta))) = -  {\textstyle\sum}_{j \in \Nmc_{i}} p_{ij}(\theta) h_{ij}(\x)   
	. 
\end{align}

It is worth to notice that \eqref{eq:secure_consensus} reproduces exactly the linear consensus protocol introduced in \eqref{eq:LAP}, since $f_i(\x,u_{i}) = -\mathrm{row}_{i}[\L] \x$, $\forall i = 1,\ldots,n$. However, a different point of view is here offered, since the adopted network manager may broadcast the codeword $\theta$ in order to redesign a subset of the edge weights whenever an external disturbance affects the integrity of the information exchanged between a couple of nodes in the network (e.g., set a specific edge weight to $0$ if it is detected to be compromised). Also, dynamics \eqref{eq:secure_consensus} shows both the presence and separation between the encryption mechanism to secure the signals sent by the network manager and the state information spreading. Indeed, defining $\p(\theta) = \vec{i=1}{n}{p_i(\theta)} \in \Rset{n^{2}}$ and $\H(\x) = \diag{i=1}{n}{h_i(\x(t))} \in \Rset{N \times n^{2}}$, dynamics \eqref{agent_dynamics}-\eqref{eq:secure_consensus} finally takes the form of
\begin{equation}\label{eq:SBDC}
	\dot{\x} = -\H(\x) \p(\theta),
\end{equation}
and, thus, the following result can be stated.
\begin{lem}\label{lem:standardandsecurepersp}
	The SBDC protocol \eqref{eq:SBDC} reaches consensus for any given objective decoding function $\p$ satisfying \eqref{eq:pijbasicchar}.
\end{lem}

\begin{proof}
	By construction, dynamics \eqref{eq:SBDC} and \eqref{eq:LAP} are equivalent. 
	Indeed, by \eqref{eq:secure_consensus}, the $i$-th equation of \eqref{eq:SBDC} can be rewritten as $\dot{x}_{i} = - {\textstyle\sum}_{j \in \Nmc_{i}} p_{ij}(\theta) h_{ij}(\x),$
	so that term $(i,j)$ in the above summation is equal to $(w_{ij}(x_{i}-x_{j}))$, if $(i,j) \in \Emc$, or it is zero, otherwise.
\end{proof}

As we will see in the next section, the benefits of such a perspective directly connect with the possibility of designing an objective coding map $\p$ hiding the information on edge weights and yielding guarantees on the robust stability of the consensus protocol \eqref{eq:SBDC}.  In particular, a codeword $\theta \in \Theta$ (when belonging to some Euclidean subspace) is deviated from its nominal value following a cyber-physical attack $\delta^{\theta}$, i.e., $(\theta + \delta^{\theta})$ is received by the function $\p$. 
Fig. \ref{fig:SBDC-block_diag.jpg} summarizes the developments obtained so far, describing the basic framework in which the next investigation is carried out.

\begin{figure}[!t]
	\centering
	\vspace{5pt}
	\includegraphics[scale=0.43]{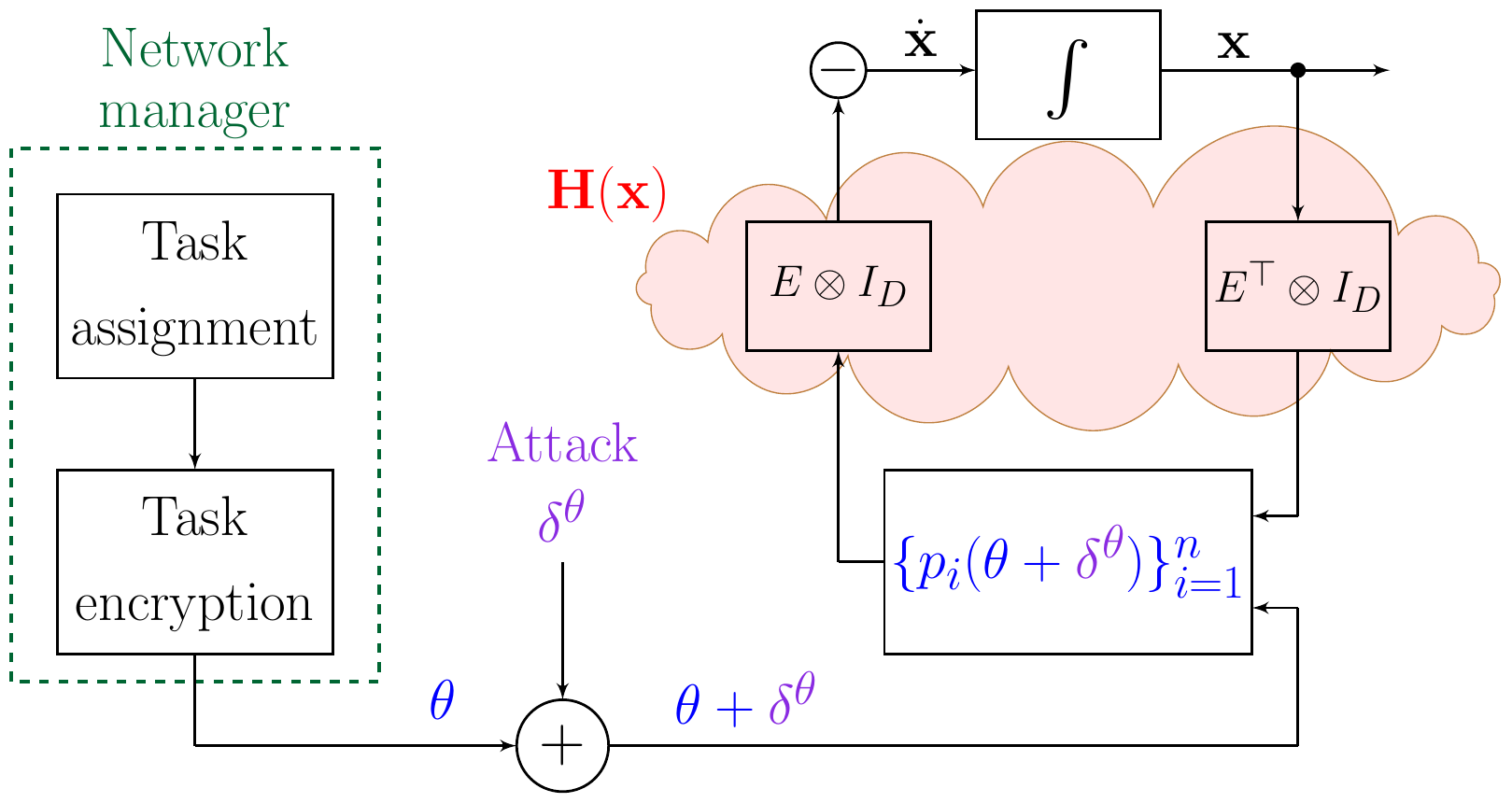}
	\caption{Block diagram depicting relation \eqref{eq:SBDC} and the presence of a cyber-physical attack $\delta^{\theta}$ deviating a sent codeword $\theta$. 
	}
	\label{fig:SBDC-block_diag.jpg}
\end{figure}

\section{Robustness to channel tampering} \label{sec:robust-analysis}

One of the goals of this study aims at the design of networks that are secure to channel tampering while accomplishing the consensus task. To this end, we propose to embed the system with security measures that allow to make it robust to small signal perturbations on a single edge. In the sequel, a description for the channel tampering is provided along with the relative robustness analysis for the devised SBDC protocol. 

\subsection{Model for the channel tampering}
This particular channel tampering problem under investigation is formulated as follows. 
Let the prescribed codeword $\theta$ be subject to a  deviation (i.e., an attack) $\delta^\theta \in \bm{\Delta}^{\theta} = \{\delta^\theta \, : \,  \|\delta^\theta \|_{\infty} \leq \bar{\delta}^\theta  \}$. 
To proceed with our analysis within a plausible framework, we let $\Theta$ be a Euclidean subspace, namely $\Theta  \subseteq \Rset{n^{2}}$, and allow a codeword $\theta  = \vec{i=1}{n}{\theta_i} \in \Theta$ to be decomposed into (at most) $n(n-1)/2$ meaningful ``subcodewords" $\theta^{(k)} := [\theta_i]_j = \theta_{ij}$, with $k=1,\ldots,m$, such that $\theta_{ij}=\theta_{ji}$, if $i\neq j$, and $\theta_{ii}$ takes an arbitrary value, for $i=1,\ldots,n$. Each $\theta_{ij} \in \Theta_{ij}  \subseteq \Rset{}$ can be seen as the $j$-th component of the $i$-th codeword piece $\theta_i$, with $i=1,\ldots,n$. Such subcodewords directly affect the value of $p_{ij}(\theta)$ if and only if $j \in \Nmc_i$, i.e., it holds that $p_{ij}(\theta) = p_{ij}(\theta_{ij}) \in \Pi_{ij} \subseteq \mathbb{R}$ for all $(i,j) \in \Emc$, with $\Pi_{ij} $ such that $\Pi_{i} = \Pi_{i1} \times \cdots \times \Pi_{ij} \times \cdots \times \Pi_{in}$.
Hence, the consensus description we account for to support this analysis is such that the $i$-th nominal dynamics in \eqref{eq:SBDC} is altered into
\begin{equation}\label{eq:perturbedSBDC}
	\dot{x}_{i} = - {\textstyle\sum}_{j \in \Nmc_{i}} p_{ij}(\theta_{ij} + \delta^{\theta}_{ij}) h_{ij}(\x), \quad i = 1,\ldots,n,
\end{equation}
with $\delta^{\theta}_{ij} = [\delta^{\theta}_i]_j $ and $\delta^{\theta}_{i}$ satisfying $\delta^{\theta}  = \vec{i=1}{n}{\delta^{\theta}_i}$. 
Therefore, in this direction, we aim to solve the following:

\begin{prb}\label{prob:resilenceSBDC}
	Find objective functions $p_{ij}$ such that \eqref{eq:perturbedSBDC} reaches consensus, independently from the codeword $\theta \in \Theta \subseteq \Rset{n^{2}}$, while the underlying MAS is subject to an attack $\delta^{\theta} \in \bm{\Delta}^{\theta}$ focused on a single edge $(u,v) \in \Emc$, i.e., with $\delta_{ij}^{\theta} = 0$ for all $(i,j)\in \Emc \setminus \{(u,v)\}$. Also, provide robustness guarantees for a given perturbation set $\bm{\Delta}^{\theta}$ in terms of the maximum allowed magnitude (denoted with $\rho_{uv}^{\theta}$) for component $\delta_{uv}^{\theta}$.
\end{prb}

\subsection{Robustness of the SBDC} \label{sec:ressbdc}
Within the setup described so far, it is possible to exploit Lem. \ref{thm:effective_resistance} and provide guarantees for the robustness of system \eqref{eq:perturbedSBDC} when the target of a cyber-physical threat is a single edge. To proceed in this way, we resort to the study of perturbations of the type $\delta_{uv}^{w} = \delta_{uv}^{w}(\theta_{uv},\delta_{uv}^{\theta})$ affecting weight $p_{uv}(\theta_{uv}) = w_{uv}$ and caused by a deviation $\delta_{uv}^{\theta}$ focused on connection $(u,v) \in \Emc$. 
Nevertheless, further assumptions on the $p_{i}$'s are required to tackle Prob. \ref{prob:resilenceSBDC}. Indeed, this robustness analysis is necessarily restricted to a particular choice for the objective coding, that is for concave and Lipschitz continuous differentiable functions $p_{i}$. More precisely, we let
the $i$-th objective coding function $p_{i}: \Theta \rightarrow \Pi_{i}$ 
adopted in model \eqref{eq:perturbedSBDC} possess the following characterization:
\begin{enumerate}
	\item[\textit{(i)}] values $[p_{i}(\theta)]_{j} = p_{ij}(\theta_{ij})$, with $\theta_{ij} = [\theta_i]_j$, satisfy \eqref{eq:pijbasicchar} for all $(i,j) \in \Emc$ and  are not constant w.r.t. $\theta_{ij}$; 
	\item[\textit{(ii)}] $p_{ij}$ is concave $\forall \theta \in \Theta$, i.e., $p_{ij}(\varsigma \eta_{1} + (1-\varsigma) \eta_{2}) \geq \varsigma p_{ij}(\eta_{1}) + (1-\varsigma) p_{ij}(\eta_{2})$, $\varsigma \in [0,1]$, $\forall \eta_{1},\eta_{2} \in \Theta_{ij}$;
	\item[\textit{(iii)}] $p_{ij}$ is Lipschitz continuous and differentiable w.r.t. $\theta$, implying $\exists K_{ij} \geq 0: ~|p_{ij}^{\prime}(\theta_{ij})| \leq K_{ij}$, $\forall (i,j) \in \Emc$.
\end{enumerate} 
While property \textit{(i)} is standard to obtain an equivalence between \eqref{eq:perturbedSBDC} in absence of attacks and its nominal version \eqref{eq:SBDC}, hypotheses \textit{(ii)}-\textit{(iii)}, demanding for concavity and Lipschitz continuity along with differentiability respectively, may not appear intelligible at a first glance.
The reason for such a characterization is clarified in the next theorem, providing the key result to solve Prob. \ref{prob:resilenceSBDC}.
\begin{thm} \label{thm:maximum_perturb_codeword}
	Assume the above characterization \textit{(i)-(iii)} for objective decoding functions $p_i$ holds.  
	Then, for an injection attack $\delta^{\theta} \in \bm{\Delta}^{\theta}$ on a single edge $(u,v)\in \Emc$, i.e., with $\delta^\theta_{ij} = 0$ for all $(i,j) \in \Emc \setminus \{(u,v)\}$, the perturbed consensus protocol \eqref{eq:perturbedSBDC} is stable for all $\delta^\theta_{uv}$ such that
	\begin{equation}\label{eq:sgrepatt}
		|\delta^\theta_{uv}| \leq \rho^\theta_{uv} = (K_{uv} \Rmc_{uv}(\Gmc) )^{-1},
	\end{equation}
	independently from the values taken by any codeword $\theta \in \Theta$. 
\end{thm}

\begin{proof}
	As the nominal system \eqref{eq:SBDC} associated to \eqref{eq:perturbedSBDC} is stable by virtue of Lem. \ref{lem:standardandsecurepersp}, characterization \textit{(i)-(iii)} determines each ordered logical step to conclude the thesis through Lem. \ref{thm:effective_resistance}. Firstly, condition \textit{(i)} is necessary to construct at least a correspondence from $\theta_{ij}$ to the weight $w_{ij}$ for all edges $(i,j) \in \Emc$. 
	Secondly, condition \textit{(ii)} expresses a concavity requirement for the $p_{ij}$'s, leading inequality $p_{ij}(\theta_{ij} + \delta^{\theta}_{ij}) \leq p_{ij}(\theta_{ij}) + p^{\prime}_{ij}(\theta_{ij})\delta^{\theta}_{ij}$ to hold for any deviation $\delta^{\theta} \in \bm{\Delta}^{\theta}$, when $p^{\prime}_{ij}(\theta_{ij})$ exists finite for all $\theta_{ij}$. Consequently, \textit{(i)} also forces $K_{ij}>0$ and \textit{(iii)} leads to
	\begin{equation}\label{eq:concavity}
		p_{ij}(\theta_{ij} + \delta^{\theta}_{ij}) -p_{ij}(\theta_{ij})  \leq  K_{ij}\delta^{\theta}_{ij}, \quad \forall (i,j) \in \Emc.
	\end{equation}
	The product $ K_{ij} \delta^{\theta}_{ij}$ in the r.h.s. of \eqref{eq:concavity} is key, as $K_{ij} |\delta^{\theta}_{ij}|$ can be seen as the maximum magnitude of an additive perturbation $\delta^{w}_{ij} := p_{ij}(\theta_{ij} + \delta^{\theta}_{ij})-p_{ij}(\theta_{ij})$ 
	affecting the nominal weight $w_{ij}=p_{ij}(\theta_{ij})$ independently from the transmitted codeword $\theta$. That is, under \textit{(i)-(iii)} model \eqref{eq:perturbedSBDC} can be reformulated as
	\begin{align}\label{agent_dynamics_uncertain}
		\dot{\x} = - \H(\x) (\p(\theta)+\delta^{w}),
	\end{align}
	where $\delta^{w} \in \bm{\Delta}^w= \{\delta^{w} \, : \,  \|\delta^w \|_{\infty} \leq \bar{\delta}^w \}$, such that $\delta^{w} = \vec{i=1}{n}{\delta_{i}^{w}}$ and $[\delta^{w}_{i}]_{j} = \delta^{w}_{ij} \leq  K_{ij} | \delta^{\theta}_{ij}|$.
	Therefore, imposing inequality $K_{uv} | \delta^{\theta}_{uv}| \leq \Rmc_{uv}(\Gmc)^{-1} $ in accordance with Lem. \ref{thm:effective_resistance} leads to the thesis, since $K_{uv} | \delta^{\theta}_{uv}|$ can be seen as an upper bound of the deviation $|\delta_{uv}^{w}| \leq K_{uv} | \delta^{\theta}_{uv}|$ for edge $(u,v) \in \Emc$ w.r.t. to an altered subcodeword $\theta_{uv}+\delta^{\theta}_{uv}$.
\end{proof}


\begin{rmk}
	It is worth to highlight that inequality \eqref{eq:sgrepatt} yields a \textit{small-gain} interpretation of the allowable edge-weight uncertainty that guarantees the network to be robustly stable within a framework where any value of a codeword $\theta \in \Theta$ is considered, provided that mapping structure \textit{(i)-(iii)} for the design of $(\theta, \p(\theta))$ is adopted.\footnote{This is the broadest setup possible, as more than a single edge weight would be altered if more complex structures of $(\theta, \p(\theta))$ were considered.}
	In addition, Thm. \ref{thm:maximum_perturb_codeword} may be conservative with regard to free-objective-coding stability margins offered by Lem. \ref{thm:effective_resistance}, since $ |\delta_{uv}^{w}| \leq K_{uv} | \delta^{\theta}_{uv}| $. 
\end{rmk}


Another critical aspect arising from Thm. \ref{thm:maximum_perturb_codeword} is reported, i.e. the tradeoff between objective coding and robustness. 
\begin{fact}\label{fact:tradeoff}
	The encoding capability of $p_{uv}$ can be expressed (locally) in terms of  the Lipschitz constant $K_{uv}$, since, given an arbitrarily small neighborhood $U_{uv}^{\theta} := [a,b] \subseteq \Theta_{uv}$ centered around the points $\theta_{uv}$ with highest absolute slope $K_{uv}$, the image subset $P_{uv}(U_{uv}^{\theta}) = [p_{uv}(a),p_{uv}(b)] \subseteq \Pi_{uv}$ dilates\footnote{Dilations are intended in terms of the Lebesgue measure of a set.} as $K_{uv}$ increases. %
	On the other hand, as $K_{uv}$ 
	decreases, the maximum magnitude $\rho^\theta_{uv}$ of admissible deviations $\delta_{uv}^{\theta}$ grows, leading to a higher robustness w.r.t edge $(u,v)$. In particular, for $K_{uv}<1$, the robustness of \eqref{eq:SBDC} is higher w.r.t. the corresponding nominal system.
	
\end{fact}

The next concluding proposition yields a deep insight to grasp the tradeoff arising from Fact \ref{fact:tradeoff} by also putting Lem. \ref{thm:effective_resistance} and Thm. \ref{thm:maximum_perturb_codeword} in comparison.
\begin{prop}\label{prop:tradeoff}
	Let $a_{uv}, b_{uv} \in \Rset{}$, $b_{uv} \neq 0$, and $(u,v) \in \Emc$ the single edge under attack. Then it holds that $U^{\theta}_{uv} = \Theta_{uv}$, $P_{uv}(U^{\theta}_{uv}) = \Pi_{uv}$ and \eqref{eq:sgrepatt} is exactly equivalent to \eqref{eq:fundamineqres}, that is $|\delta_{uv}^{w}|=K_{uv} | \delta^{\theta}_{uv}| $, if and only if $p_{uv}(\eta) = b_{uv} \eta +a_{uv}$.
\end{prop}
\begin{proof}
	As $p_{uv}(\eta) = b_{uv} \eta + a_{uv}$, all points $\eta \in \Theta_{uv}$ have the same absolute slope $K_{uv} = |b_{uv}|$, thus 
	implying $P_{uv}(\Theta_{uv}) = \Pi_{uv}$. Also, condition $p_{uv}(\eta) = b_{uv} \eta + a_{uv}$, with $b_{uv} \neq 0$, is sufficient and necessary to obtain $|\delta_{uv}^{w}|  = K_{uv}|\delta_{uv}^{\theta}|$ for all $ \eta \in \Theta_{uv}$, since \eqref{eq:concavity} applied to edge $(u,v)$ holds with the equality. 
\end{proof}

Prop. \ref{prop:tradeoff} shows the unique scenario where the tradeoff in Fact \ref{fact:tradeoff} holds strictly, namely  it holds globally $\forall \eta \in \Theta_{uv}$, also allowing \eqref{eq:sgrepatt} not to be conservative\footnote{In other words, the conservatism expressed in Thm. \ref{thm:maximum_perturb_codeword} arises only if decoding functions that are nonlinear in their argument are adopted.} w.r.t. \eqref{eq:fundamineqres}.


\section{An application to decentralized estimation} \label{sec:PI-ACE}
Decentralized estimation and control of graph connectivity for mobile sensor networks is often required in practical applications \cite{ZelazoFranchiBulthoff2015, YangFreemanGordon2010}. 
As outlined in \cite{YangFreemanGordon2010}, the Fiedler eigenvalue and eigenvector of a graph can be estimated in a distributed fashion by employing the so-called \textit{decentralized power iteration algorithm} (DPIA) with a uniformly weighted \textit{PI average consensus estimator} (PI-ACE). In this setup, $n$ agents measure a time-varying scalar $c_i=c_i(t)$, and by communication over an undirected and connected graph estimate the average of the signal, $\hat{c}(t) = n^{-1} \sum_{i=1}^{n} c_{i}(t)$. By considering estimation variables $y_{i} = y_{i}(t) \in \mathbb{R}$ and $q_{i} = q_{i}(t) \in \mathbb{R}$, $i = 1, \ldots, n$, the continuous-time estimation dynamics in question associated to the $i$-th agent is given by \cite{YangFreemanGordon2010}
\begin{equation}\label{eq:PIaverageconsensusestimator}
	\begin{cases}
		\dot{y}_{i} = \alpha(c_{i}-y_{i}) - K_{P} \sum\limits_{j \in \Nmc_{i}} (y_{i}-y_{j}) +K_{I} \sum\limits_{j \in \Nmc_{i}} (q_{i}-q_{j}) \\
		\dot{q}_{i} = -K_{I} \sum\limits_{j \in \Nmc_{i}} (y_{i}-y_{j})
	\end{cases},
\end{equation}
where $\alpha > 0$ represents the rate new information replaces old information and $K_{P}$, $K_{I} >0$ are the PI estimator gains. Remarkably, the latter constants play an important role in the convergence rate of estimator \eqref{eq:PIaverageconsensusestimator}, as the the estimation dynamics is demanded to converge fast enough to provide a good approximation of $\hat{c}=\hat{c}(t)$ (which is determined by each component of $y$, i.e. $\lim_{t \rightarrow \infty} |\hat{c}(t)-y_{i}(t)| = 0$ for $i=1,\ldots,n$ is desired).
In the sequel, we thus firstly provide a spectral characterization pertaining such an estimator dynamics and then we adapt the results obtained in Sec. \ref{sec:robust-analysis} to this specific framework, finally illustrating the criticalities of the DPIA.

\subsection{On the spectral properties of the PI-ACE}

Setting $y = \begin{bmatrix}
	y_{1} & \cdots & y_{n}
\end{bmatrix}^{\top}$, $q = \begin{bmatrix}
	q_{1} & \cdots & q_{n}
\end{bmatrix}^{\top}$ and $x = \begin{bmatrix}
	y^{\top} & q^{\top}
\end{bmatrix}^{\top}$, $\c = \begin{bmatrix}
	\alpha c^{\top} & \zerovec{n}^{\top}
\end{bmatrix}^{\top}$, dynamics \eqref{eq:PIaverageconsensusestimator} can be also rewritten as
\begin{equation}\label{eq:PIaverageconsensusestimator_}
	\dot{x} = -M x + \c,
\end{equation}
such that
\begin{equation}\label{eq:M}
	M = \begin{bmatrix}
		K_{P}L+\alpha I_{n} & -K_{I}L \\
		K_{I}L & \zerovec{n \times n}
	\end{bmatrix},
\end{equation}
where, throughout all this section, $L$ stands for the unweighted graph Laplacian associated to the unweighted network $\Gmc_{0} = (\Vmc, \Emc, \Wmc_{0})$, $\Wmc_{0} = \{1\}_{k=1}^{m}$.
Clearly, \eqref{eq:PIaverageconsensusestimator_} can be thought as a driven second-order consensus dynamics whose stability properties depend on the eigenvalues $\lambda_{l}^{M}$, $l= 1,\ldots, 2n$, of state matrix $M$. In this direction, we characterize the eigenvalues of $M$ in function of those of $L$ by means of the following proposition to grasp an essential understanding of the convergence behavior taken by dynamics \eqref{eq:PIaverageconsensusestimator_}. 

\begin{prop}\label{Prop:eigM}
	The eigenvalues of matrix $M$, defined as in \eqref{eq:M}, are given by
	\begin{equation}\label{eq:eigsofMinfuncofeigsofL}
		\lambda_{2(i-1)+j}^{M} = \varphi_{i} + (-1)^{j} \sigma_{i}, \quad i=1,\ldots,n, ~\forall j\in \{1,2\},
	\end{equation}
	where
	\begin{equation}\label{eq:eigsofMinfuncofeigsofL_particular}
		\begin{cases}
			\varphi_{i} = (\alpha+K_{P}\lambda^{L}_{i})/2\\
			\sigma_{i} = \sqrt{\varphi_{i}^{2}-(K_{I}\lambda^{L}_{i})^{2}}, \quad \text{s.t. } \Im[\sigma_{i}] \geq 0
		\end{cases}.
	\end{equation}
	Furthermore, $\lambda_{1}^{M} = 0$ and $\Re [\lambda_{l}^{M}] > 0$ for $l = 2,\ldots,2n$.
\end{prop}

The proof of Prop. \ref{Prop:eigM} can be found in Appendix and, for a further discussion on the convergence properties of system \eqref{eq:PIaverageconsensusestimator_} and the estimation of signal $\hat{c}(t)$, the reader is referred to \cite{YangFreemanGordon2010,FreemanYangLynch2006}. 
In fact, in the sequel, we aim at the adaptation of theoretical results obtained in Sec.  \ref{sec:robust-analysis} to this specific framework. Considering that $K_{P}$, $K_{I}$ and $\alpha$ can be seen as parameters to be sent by the network manager, it is, indeed, possible to discuss the following relevant practical scenario. 

\subsection{Application scenario}
We now consider an application scenario with a couple of setups based on the perturbed second-order consensus protocol
\begin{equation}\label{eq:perturbedSBDC_secondorder}
	\begin{cases}
		\dot{y}_{i} &\hspace{-.2cm}=  p_{ij}^{(\alpha)}(\theta_{ij}+ \delta^{\theta}_{ij})(c_{i}-y_{i})-  {\underset{{j \in \Nmc_{i}}}{\sum}}  p_{ij}^{(K_{P})}(\theta_{ij} + \delta^{\theta}_{ij}) h_{ij}(y) \\
		&+\underset{j \in \Nmc_{i}}{\sum} p_{ij}^{(K_{I})}(\theta_{ij} + \delta^{\theta}_{ij})  h_{ij}(q)  \\ 
		\dot{q}_{i} &\hspace{-.2cm}=  - {\underset{j \in \Nmc_{i}}{\sum}} p_{ij}^{(K_{I})}(\theta_{ij} + \delta^{\theta}_{ij}) h_{ij}(y),  
	\end{cases}
\end{equation}
and defined through decoding functions and information localization functions
\begin{align}\label{eq:p_secondorder}
	p_{ij}^{(\varsigma)}(\theta_{ij}) &= \begin{cases}
		\varsigma \quad\quad \forall (i,j)\in \Emc;\\  
		0 \quad~~~ \text{otherwise};
	\end{cases} \\
	h_{ij}(\omega) &= \begin{cases}
		\omega_{i} - \omega_{j} \quad \forall (i,j)\in \Emc;\\
		0 \quad~~~~~~~ \text{otherwise}.
	\end{cases} \label{eq:h_secondorder}
\end{align}

In the first setup, named $S1$, we assume that a perturbation over a single codeword affects parameter $K_{P}$, thus changing quantities $\varphi_{i}$. Also, we suppose that gains $\alpha$, $K_{I}$ are not perturbed and are correctly received (or already known) by all agents in the network $\Gmc$.\\ 

It is worth to note that all the results on robustness given so far are directed towards the preservation of the positive semi-definiteness of the weighted Laplacian matrix, which is also related to the stability of the corresponding consensus protocol. In particular, in this application, terms $(K_{P}\lambda^{L}_{i})$ can be thought as eigenvalues of the weighted Laplacian $L_{P} = K_{P}EE^{\top}$. In addition, as the proof of Prop. \ref{Prop:eigM} reveals, since $\varphi_{i} > 0$ for all $i=1,...,n$ then $\Re [\lambda_{l}^{M}] > 0$ for all $l = 2,\ldots,2n$ is ensured. Hence, as far as the perturbed values of $\varphi_{i}$, $i=1,...,n$, remain strictly positive for any value of $\alpha>0$ then stability for a perturbed version of protocol \eqref{eq:PIaverageconsensusestimator_} can be guaranteed, since each $\varphi_{i}$ can also be seen as an eigenvalue of matrix $M_{P} = (\alpha I_{n} +L_{P})/2$. Indeed, the worst case in this setup arises when $\alpha$ is arbitrarily small, implying that the stability of \eqref{eq:PIaverageconsensusestimator_} can be guaranteed if $L_{P}$ preserves its positive semidefiniteness under attack. Consequently, inequality \eqref{eq:sgrepatt} can be applied to this setup, accounting for an auxiliary graph $\Gmc_{P}$ constructed from $L_{P}$, whenever a single edge codeword associated to weight $K_{P}$ is perturbed. This reasoning is better formalized in the following concluding corollary.

\begin{cor} \label{thm:maximum_perturb_codeword_second_KP}
	Assume the characterization \textit{(i)-(iii)} in Sec. \ref{sec:ressbdc} holds for objective decoding functions $p_i$. Let $\omega \in \mathbb{R}^{n}$, $\varsigma \in \mathbb{R}$ and $\Gmc_{P} = (\Vmc,\Emc,\Wmc_{P})$, with $\Wmc_{P} = \{K_{P}\}_{k=1}^{m}$, be a graph constructed from $L_{P} = K_{P} EE^{\top}$, given $K_{P} > 0$.
	Then, for an injection attack $\delta^{\theta} = \begin{bmatrix}
		\delta^{\theta\top}_{\alpha} & \delta^{\theta\top}_{K_{P}} & \delta^{\theta\top}_{K_{I}}
	\end{bmatrix}^{\top} = \begin{bmatrix}
		\zerovec{n^{2}}^{\top} & \delta^{\theta\top}_{K_{P}} & \zerovec{n^{2}}^{\top}
	\end{bmatrix}^{\top}$, $\delta^{\theta}_{K_{P}} \in \bm{\Delta}^{\theta}$, on a single edge $(u,v)\in \Emc$, i.e., with $\delta^\theta_{K_{P},ij} = 0$ for all $(i,j) \in \Emc \setminus \{(u,v)\}$, protocol \eqref{eq:perturbedSBDC_secondorder}-\eqref{eq:p_secondorder}-\eqref{eq:h_secondorder}
	is stable for all $\alpha , K_{P} , K_{I} > 0$ and $\delta^\theta_{uv}$ such that
	\begin{equation}\label{eq:sgrepatt_second_KP}
		|\delta^\theta_{uv}| \leq \rho_{P,uv}^{\theta} = (K_{uv} \Rmc_{uv}(\Gmc_{P}) )^{-1},
	\end{equation}
	independently from the values taken by any codeword $\theta = \begin{bmatrix}
		\theta_{\alpha}^{\top} & \theta_{K_{P}}^{\top} & \theta_{K_{I}}^{\top}
	\end{bmatrix}^{\top} \in \Theta \subseteq \mathbb{R}^{3n^{2}}$.
\end{cor}

\begin{proof}
	The result is a direct consequence of Prop. \ref{Prop:eigM} applied to Thm. \ref{thm:maximum_perturb_codeword} within setup $S1$, which is characterized by \eqref{eq:perturbedSBDC_secondorder}-\eqref{eq:p_secondorder}-\eqref{eq:h_secondorder}.
\end{proof}

In the second setup, named $S2$, we differently assume that only three scalar subcodewords $\theta_{\alpha}$, $\theta_{K_{P}}$ and $\theta_{K_{I}}$, constituting codeword $\theta = \begin{bmatrix}
	\theta_{\alpha} & \theta_{K_{P}} & \theta_{K_{I}}
\end{bmatrix}^{\top} \in \Theta \subseteq \mathbb{R}^{3}$, are broadcast by the network manager. This framework can be motivated by the attempt to reduce computational burden, network complexity or overall energy consumption. Each agent $i$ then receives $\theta$ and uses three decoding functions $p_{ij}^{(\alpha)}(\theta_{ij}) = p^{(\alpha)}(\theta_{\alpha})$, $p_{ij}^{(K_{P})}(\theta_{ij}) = p^{(K_{P})}(\theta_{K_{P}})$, $p_{ij}^{(K_{I})}(\theta_{ij}) = p^{(K_{I})}(\theta_{K_{I}})$ for all $(i,j) \in \Emc$ to unveil the weights $\alpha$, $K_{P}$, $K_{I}$ encoded in $\theta_{\alpha}$, $\theta_{K_{P}}$, $\theta_{K_{I}}$, respectively.\\ 
With such a preliminary description for $S2$, we now provide the following robust consensus guarantee.

\begin{thm}\label{thm:secordestallpert}
	Assume the characterization \textit{(i)-(iii)} in Sec. \ref{sec:ressbdc} holds for objective decoding functions $p^{(\alpha)}$, $p^{(K_{P})}$, $p^{(K_{I})}$ with Lipschitz constants $K_{\alpha} , K_{K_{P}} , K_{K_{I}} > 0$, respectively. 
	Let $\delta^{\theta} = \begin{bmatrix}
		\delta^{\theta}_{\alpha} & \delta^{\theta}_{K_{P}} & \delta^{\theta}_{K_{I}}
	\end{bmatrix}^{\top} $, with $\delta^{\theta}_{\alpha},\delta^{\theta}_{K_{P}},\delta^{\theta}_{K_{I}} \in \BDelta^{\theta}$ scalar time-varying perturbations, be the an injection attack affecting all the edges in the network.
	Then, the perturbed consensus protocol \eqref{eq:perturbedSBDC_secondorder}-\eqref{eq:p_secondorder}-\eqref{eq:h_secondorder}
	reaches agreement for all $\alpha , K_{P} , K_{I} > 0$ and $ \delta^{\theta}_{\alpha},\delta^{\theta}_{K_{P}},\delta^{\theta}_{K_{I}} $ such that
	\begin{equation}\label{eq:stabconds_cons_3sc}
		\begin{cases}
			|\delta^{\theta}_{\alpha}| < K_{\alpha}^{-1} \alpha \\
			|\delta^{\theta}_{K_{P}}| < (\lambda_{n}^{L} K_{K_{P}})^{-1} (\alpha - K_{\alpha}|\delta^{\theta}_{\alpha}|+\lambda_{n}^{L} K_{P}) \\
			|\delta^{\theta}_{K_{I}}| < K_{K_{I}}^{-1} K_{I}
		\end{cases},
	\end{equation}
	independently from the values taken by any codeword $\theta = \begin{bmatrix}
		\theta_{\alpha} & \theta_{K_{P}} & \theta_{K_{I}} 
	\end{bmatrix}^{\top} \in \Theta \subseteq \mathbb{R}^{3}$.
\end{thm}
\begin{proof}
	Recalling expressions \eqref{eq:eigsofMinfuncofeigsofL}-\eqref{eq:eigsofMinfuncofeigsofL_particular} for the eigenvalues of update matrix $M$ in \eqref{eq:M} that determines the nominal\footnote{Note that nominal dynamics \eqref{eq:PIaverageconsensusestimator_} can be obtained from \eqref{eq:perturbedSBDC_secondorder} when $\delta^{\theta}_{\alpha}=0$, $\delta^{\theta}_{K_{P}}=0$, $\delta^{\theta}_{K_{I}}=0$.} dynamics \eqref{eq:PIaverageconsensusestimator_} from Prop. \ref{Prop:eigM}, it is possible to compute the expression for the perturbed eigenvalues associated to dynamics \eqref{eq:perturbedSBDC_secondorder}. More precisely, expression \eqref{eq:eigsofMinfuncofeigsofL_particular} can be modified in function of variations $ \delta^{w}_{\alpha} = p^{(\alpha)}(\theta_{\alpha} + \delta^{\theta}_{\alpha})-\alpha$, $\delta^{w}_{K_{P}} = p^{(K_{P})}(\theta_{K_{P}} + \delta^{\theta}_{K_{P}})-K_{P}$, $\delta^{w}_{K_{I}} = p^{(K_{I})}(\theta_{K_{I}} + \delta^{\theta}_{K_{I}})-K_{I}$ as
	\begin{equation}\label{eq:phisigmMbar}
		\begin{cases}
			\overline{\varphi}_{i} = (\alpha+\delta^{w}_{\alpha}+(K_{P}+\delta^{w}_{K_{P}})\lambda^{L}_{i})/2\\
			\overline{\sigma}_{i}  = \sqrt{\overline{\varphi}_{i}^{2}-((K_{I}+\delta^{w}_{K_{I}})\lambda^{L}_{i})^{2}}, 
			~~\text{s.t. } \Im[\overline{\sigma}_{i}] \geq 0
		\end{cases}
	\end{equation}
	to find out the eigenvalues $\lambda_{2(i-1)+j}^{\overline{M}} = \overline{\varphi}_{i} + (-1)^{j} \overline{\sigma}_{i}$,  $i=1,\ldots,n, ~\forall j\in \{1,2\}$, of the update matrix $\overline{M}$ regulating dynamics \eqref{eq:perturbedSBDC_secondorder}, whose form is yielded by
	\begin{equation*}\label{eq:Mbar}
		\overline{M} = \begin{bmatrix}
			(K_{P}+\delta^{w}_{K_{P}})L+(\alpha+\delta^{w}_{\alpha}) I_{n} & -(K_{I}+\delta^{w}_{K_{I}})L \\
			(K_{I}+\delta^{w}_{K_{I}})L & \zerovec{n \times n}
		\end{bmatrix}.
	\end{equation*}
	
	It is now possible to focus on the computation of the maximum magnitude allowed for deviations $\delta^{w}_{\alpha}$, $\delta^{w}_{K_{P}}$, $\delta^{w}_{K_{I}}$. \\
	In particular, the first step to guarantee robust consensus is to ensure that $\overline{\varphi}_{i} > 0$ for all $i=1,\ldots,n$. Remarkably, the first two conditions in \eqref{eq:stabconds_cons_3sc} serve this purpose as the following reasoning holds. For all $i=1,\ldots,n$, $\overline{\varphi}_{i} > 0$ is verified if $|\delta^{w}_{\alpha}+\lambda_{i}^{L}\delta^{w}_{K_{P}}| < \alpha + \lambda_{i}^{L} K_{P}  $. By the triangle inequality, the latter condition can be replaced by $|\delta^{w}_{\alpha}|+\lambda_{i}^{L}|\delta^{w}_{K_{P}}| < \alpha + \lambda_{i}^{L} K_{P}$. Hence, exploiting the ascending magnitude of $\lambda_{i}^{L}$ w.r.t. index $i \in \{1,\ldots,n\}$, conditions $|\delta^{w}_{\alpha}| < \alpha$ and $|\delta_{\alpha}^{w}|+\lambda_{i}^{L} |\delta^{w}_{K_{P}}| < \alpha+ \lambda_{i}^{L} K_{P}$ can be imposed simultaneously by respectively looking at cases $i=1$ and $i \in \{2,\ldots,n\}$. Consequently, leveraging the concavity of functions $p^{(\alpha)}$ and $p^{(K_{P})}$ as in \eqref{eq:concavity}, namely employing $|\delta^{w}_{\alpha}| \leq K_{\alpha} |\delta_{\alpha}^{\theta}|$ and $|\delta^{w}_{K_{P}}| \leq K_{K_{P}} |\delta_{K_{P}}^{\theta}|$, the first two conditions in \eqref{eq:stabconds_cons_3sc} can be finally enforced. As a further observation, it is worth to notice that input $\overline{\c} = \begin{bmatrix}
		p^{(\alpha)}(\theta_{\alpha} + \delta^{\theta}_{\alpha})c^{\top} & \zerovec{n}^{\top}
	\end{bmatrix}^{\top}$ corresponding to system \eqref{eq:perturbedSBDC_secondorder} still remains well-defined in its sign, as $p^{(\alpha)}(\theta_{\alpha} + \delta^{\theta}_{\alpha})>0$ if first condition in \eqref{eq:stabconds_cons_3sc} holds.
	
	On the other hand, robust consensus can be guaranteed only by also ensuring that $\overline{\sigma}_{i} \neq \overline{\varphi}_{i}$ for $i = 2,\ldots,n$, so that $\overline{M}$ is prevented to have more than one eigenvalue at zero, as eigenvalue $\lambda_{1}^{\overline{M}} = 0$ is attained for any perturbation $ \delta^{\theta}_{\alpha},\delta^{\theta}_{K_{P}},\delta^{\theta}_{K_{I}}$. In this direction, only deviations $\delta^{w}_{K_{I}}$ to parameter $K_{I}$ such that $|\delta^{w}_{K_{I}}| < K_{I}$ can be accepted (see the structure of $\overline{\sigma}_{i}$ in \eqref{eq:phisigmMbar}). Exploiting again concavity, namely $|\delta^{w}_{K_{I}}| \leq K_{K_{I}} |\delta_{K_{I}}^{\theta}|$, the third condition in \eqref{eq:stabconds_cons_3sc} is lastly enforced as well.
\end{proof}

Security guarantees in \eqref{eq:stabconds_cons_3sc} are conservative, in general. Nevertheless, it is possible to find a sharp upper bound for any perturbations $ \delta^{\theta}_{\alpha},\delta^{\theta}_{K_{P}},\delta^{\theta}_{K_{I}}$ in Thm. \ref{thm:secordestallpert} if
decoding functions $p^{(\alpha)}$, $p^{(K_{P})}$, $p^{(K_{I})}$ are taken linear w.r.t. to their subcodeword arguments, similarly to $p_{uv}$ in Prop. \ref{prop:tradeoff}. 
Lastly, it is worth noticing that the second inequality in \eqref{eq:stabconds_cons_3sc} can be generalized for any admissible $\delta^{\theta}_{\alpha}$, with $|\delta^{\theta}_{\alpha}| < K_{\alpha}^{-1} \alpha$, so that any $\delta^{\theta}_{K_{P}}$ such that $|\delta^{\theta}_{K_{P}}| < K_{K_{P}}^{-1} K_{P}$ be acceptable, implying that any self-loop value $\alpha>0$ contributes to increase robust agreement.

\begin{figure*}[t!] 
	\centering
	\subfigure[Chosen topology $\Gmc$.]{\includegraphics[height=0.16\textwidth, trim={4cm 3cm 7cm 2.5cm},clip]{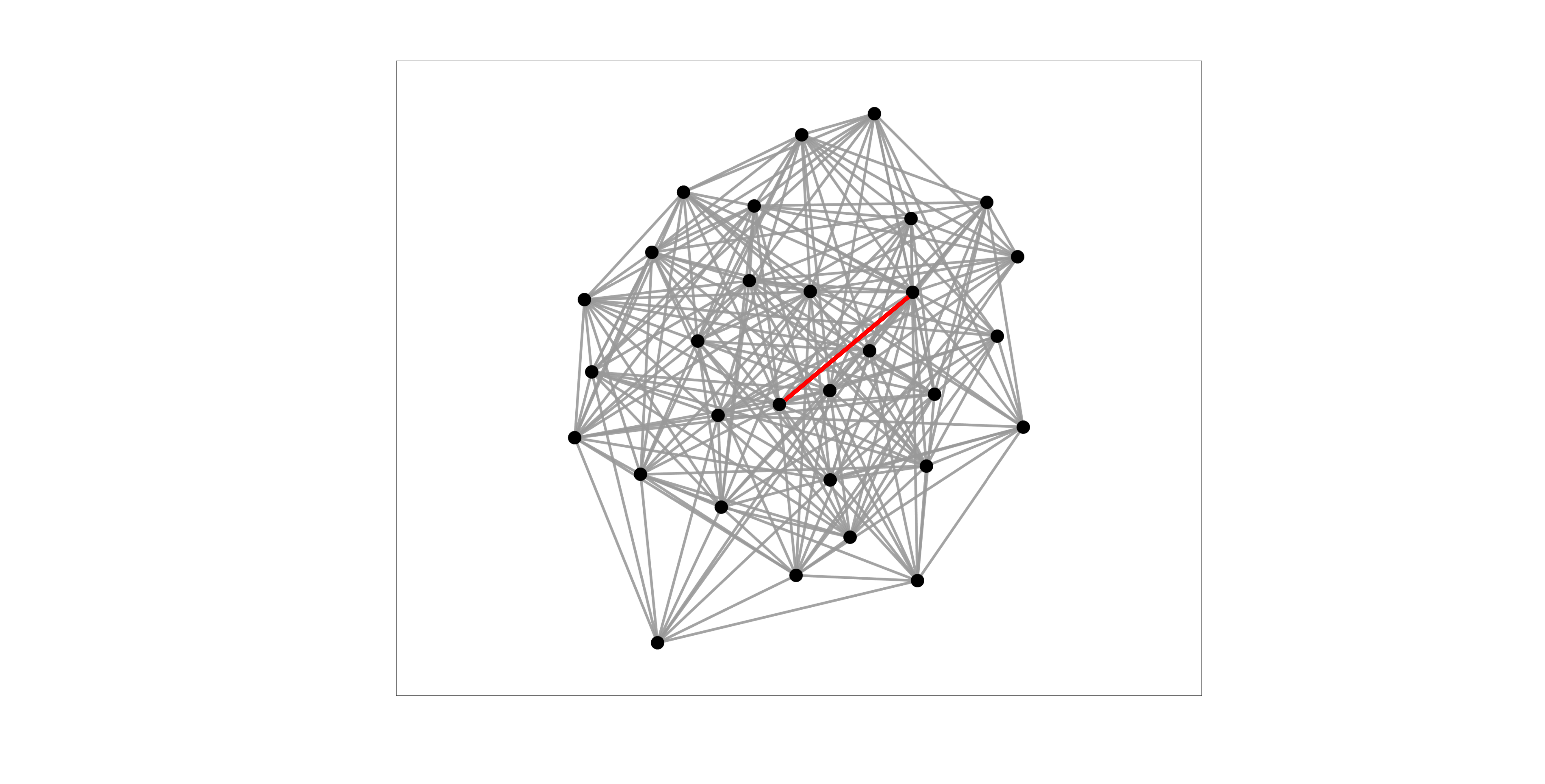}\label{fig:graph30nodes}} 
	\subfigure[\textit{Setup} $S1$: attack on red edge of $\Gmc_{P} \sim K_{P}\Gmc$ involving $K_{P}$ only.]{\includegraphics[height=0.15\textwidth, trim={3cm 0cm 2cm 2.3cm},clip]{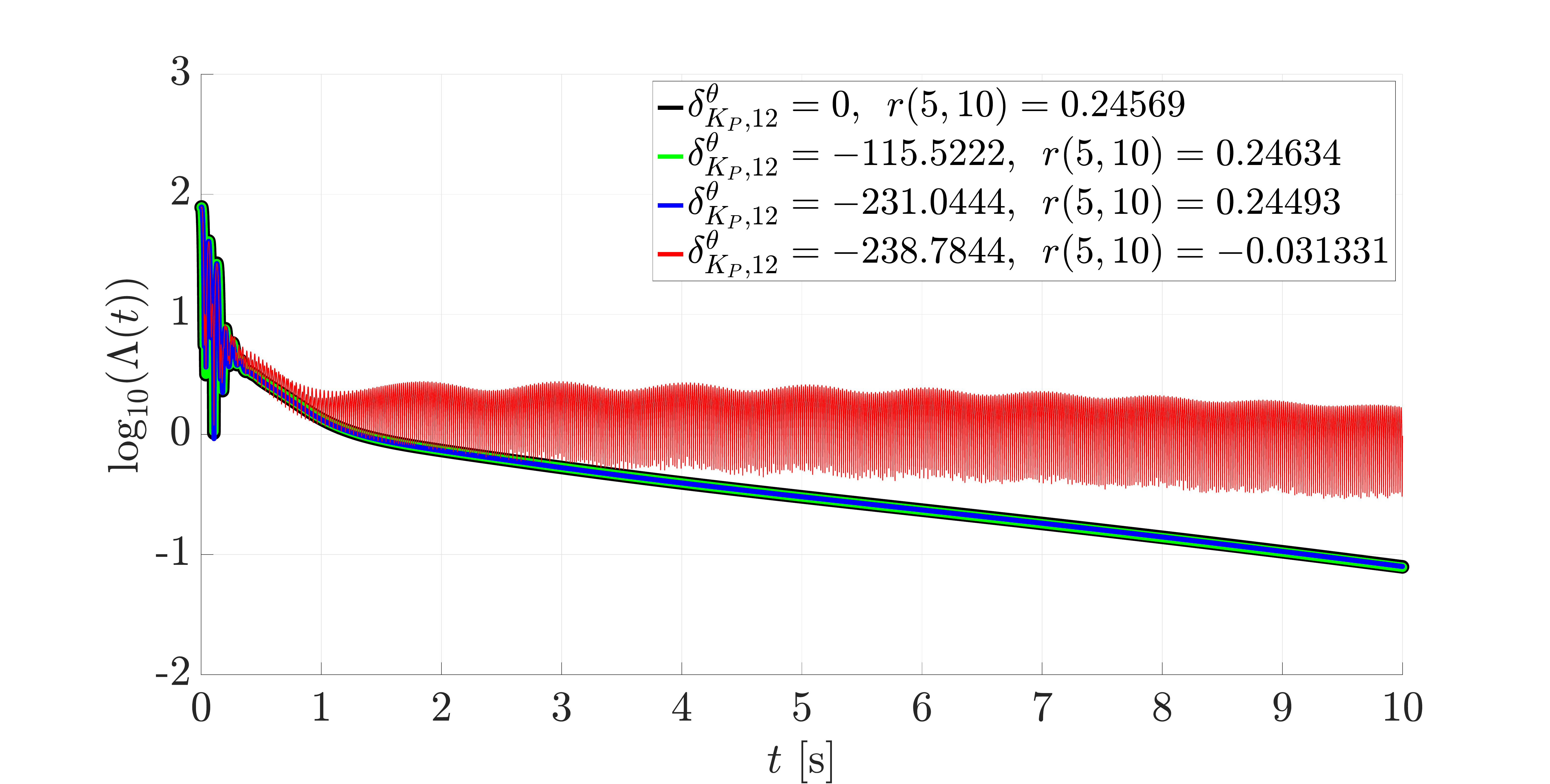}\label{fig:Esterr1}} \hspace{0.5cm}
	\subfigure[\textit{Setup} $S2$: attack on $\Gmc_{\star}$ involving all parameters $\alpha$, $K_{P}$, $K_{I}$.]{\includegraphics[height=0.15\textwidth, trim={3cm 0cm 2cm 2.3cm},clip]{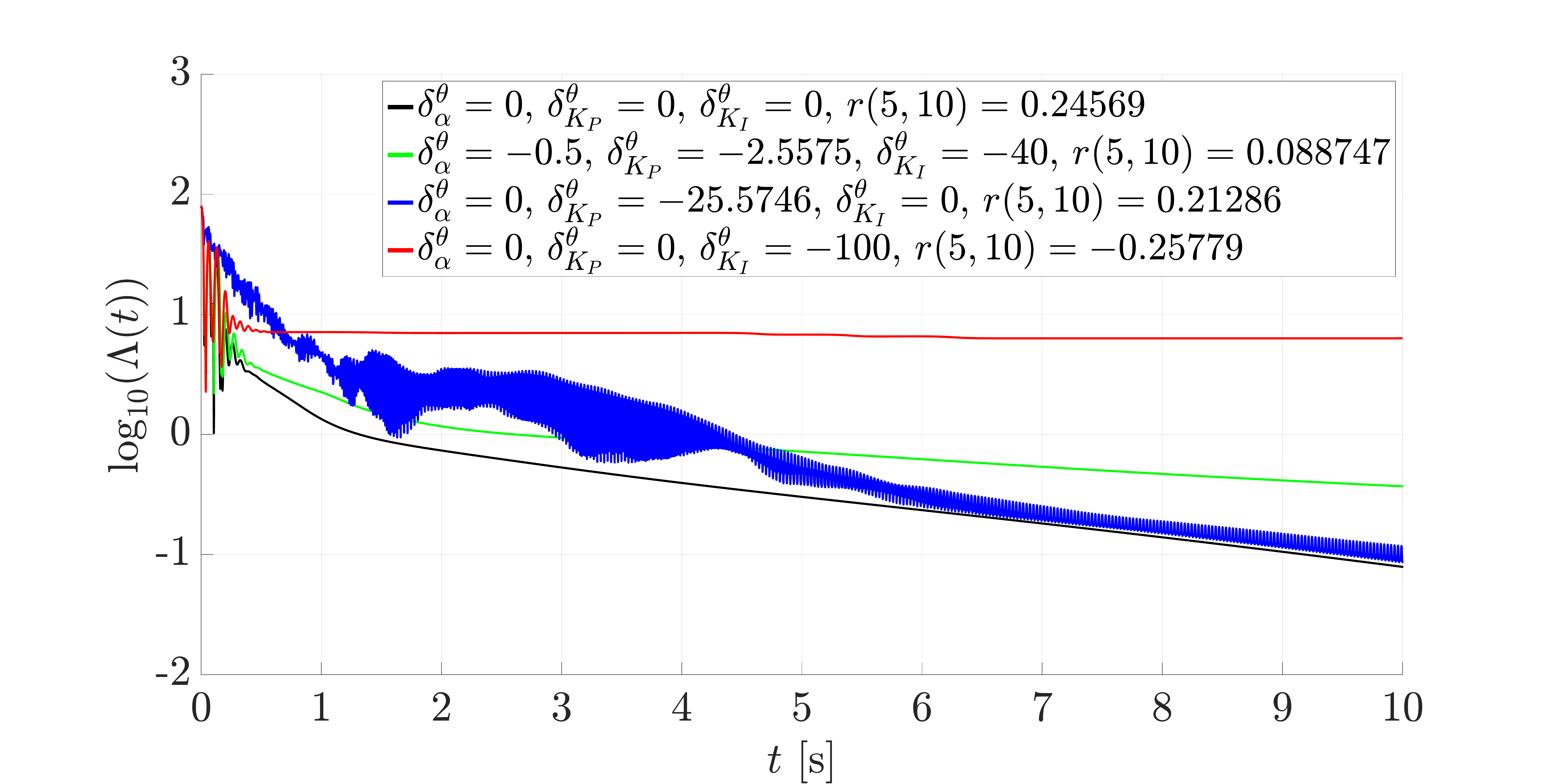}\label{fig:Esterr2}}
	\caption{Numerical results obtained from the application of the SBDC approach to the DPIA.}
	\label{fig:simulationsLambda2Est}
\end{figure*}

\subsection{Numerical examples on the DPIA criticalities}

The following numerical simulations show the secure estimation of eigenvalue $\lambda_{2}^{L} \simeq 8.6231$ of the Laplacian matrix $L$ associated to graph $\Gmc = (\Vmc, \Emc, \{1\}_{k=1}^{m})$, with $n=30$ nodes, in Fig. \ref{fig:graph30nodes}. This computation occurs in a distributed way within each agent $i\in \{1,\ldots,n\}$ and is carried out accounting for the additional dynamics\footnote{The initial conditions are selected according to a uniformly random vector with components in $(0,1)$.}
\begin{equation}\label{eq:l2est}
	\dot{\zeta}_{i} = -k_{1}y_{i,1} - k_{2} \textstyle{\sum_{j \in \Nmc_{i}}} (\zeta_{i}-\zeta_{j}) -k_{3} y_{i,2} \zeta_{i},
\end{equation}
in which $y^{(1)}=\begin{bmatrix}
	y_{1,1} & \cdots & y_{n,1}
\end{bmatrix}^{\top}$ and $y^{(2)}=\begin{bmatrix}
	y_{1,2} & \cdots & y_{n,2}
\end{bmatrix}^{\top}$ are the $y$ states of two distinct PI-ACEs of the form \ref{eq:PIaverageconsensusestimator}. In addition, the latter estimators are designed so that inputs $c_{i,1} = \zeta_{i}$ and $c_{i,2} = \zeta_{i}^{2}$ feed their dynamics. The DPIA is therefore constituted by such a system interconnection between \eqref{eq:l2est} and a couple of PI-ACEs \eqref{eq:PIaverageconsensusestimator}.\\
In the sequel, we employ network $\Gmc$ within the two setups $S1$ and $S2$ described in the previous subsections. Throughout all the discussion we assume that the nominal parameters and decoding functions are given by $\alpha = 25$, $K_{P} = 50$, $K_{I} = 10$ and $p^{(\alpha)}(\eta) = 5 \eta$, $p^{(K_{P})}(\eta) = 2 \eta$, $p^{(K_{I})}(\eta) = 0.1 \eta$, with $\eta \in \mathbb{R}$. The latter quantities are subject to numerical deviations for both the PI-ACEs associated to $y^{(1)}$ and $y^{(2)}$. Moreover, we assume that parameters $k_{1} = 60$, $k_{2} = 1$, $k_{3} = 200$ are fixed (according to requirements in \cite{YangFreemanGordon2010}) and are not affected by any type of uncertainty.\\
The $i$-th estimate $\hat{\lambda}^{L}_{2,i}$ of eigenvalue $\lambda_{2}^{L}$ can be obtained as $\hat{\lambda}^{L}_{2,i} = \lim_{t \rightarrow \infty} \lambda^{L}_{2,i}(t) $, where $\lambda^{L}_{2,i}(t) = k_{2}^{-1}k_{3}(1-y_{i,2}(t))$. We thus measure the performance of the DPIA through error $\Lambda(t) = n^{-1}\sum_{i=1}^{n} |\lambda_{2}^{L}-\lambda^{L}_{2,i}(t)|$. We also define the convergence rate $r(T_{0},T) = -(l_{T}-l_{T_{0}}+1)^{-1}\sum_{l=l_{T_{0}}}^{l_{T}} \log(\Lambda(t_{l}))/t_{l}$ that approximates the exponential decay of $\Lambda(t_{l})$, where $t_{l}$ is the discretized time stamp used by the solver and $l_{T_{0}}$, $l_{T}$ are the indexes addressing instants $T_{0}>0$, $T\geq T_{0}$, respectively. Whenever $r(T_{0},T) \leq 0$ no decay is attained over $[T_{0},T]$.

Fig. \ref{fig:Esterr1} depicts four cases wherein a constant attack $\delta_{K_{P},12}^{\theta}$ strikes edge $(1,2)$, highlighted in red, of the uniformly $K_{P}$-weighted version of $\Gmc$, namely $\Gmc_{P} = (\Vmc,\Emc,\{K_{P}\}_{k=1}^{m}) \sim K_{P}\Gmc$, according to $S1$. In this setup, the maximum allowed perturbation related to edge $(1,2)$ is given by $\rho_{12}^{\theta} = 231.0444$ (see \eqref{eq:sgrepatt_second_KP}). It can be appreciated that perturbations to subcodewords concerning $K_{P}$ do not affect the convergence rate, as far as the DPIA dynamics remain stable. Furthermore, it is worth noticing that security guarantees hold, as expected, and estimation instability certainly occurs if $\delta_{K_{P},12}^{\theta} \leq -1.0335 \rho_{12}^{\theta}$.

Considering instead $S2$, Fig. \ref{fig:Esterr2} refers to four structured constant attacks striking all the three subcodewords $\theta_{\alpha}$, $\theta_{K_{P}}$, $\theta_{K_{I}}$ broadcast by the network manager, wherein $\Gmc_{\star} = (\Vmc,\Emc,\{\star\}_{k=1}^{m}) \sim \star \Gmc$ denotes the weighted version of $\Gmc$ in Fig. \ref{fig:graph30nodes} by $\star \in \{\alpha, K_{P},K_{I}\}$. Each maximum allowed perturbation is yielded by $|\delta_{\alpha}^{\theta}| < 5$, $|\delta_{K_{P}}^{\theta}| < 1.5746-0.1149|\delta_{\alpha}^{\theta}|$ and $|\delta_{K_{I}}^{\theta}| < 100$ through \eqref{eq:stabconds_cons_3sc}. In this illustration, it is worth to observe all the different effects due to deviations for such parameters, resulting in a slowdown of the convergence rate (i.e. a decrease of $r(T_{0},T)$) or in a change to an undesired highly oscillatory behavior for the performance index $\Lambda(t)$. In particular, perturbations focusing on $\theta_{\alpha}$, $\theta_{K_{P}}$ and $\theta_{K_{I}}$ lead to slower convergence, noisy/ oscillatory estimation behavior and a considerable steady state estimation error, respectively. Furthermore, all the stability behaviors of the curves here reported comply with security guarantees in \eqref{eq:stabconds_cons_3sc}, as expected, in a non-conservative fashion (i.e. multiple zero eigenvalues appear in $\overline{M}$ for critical values of perturbations). Remarkably, the introduction of performance index $r(T_{0},T)$ is also justified by the fact that it captures the general tendency of the convergence rate for the DPIA to increase as $\lambda_{2}^{L}$ grows. Fig. \ref{fig:lambda2vsr} illustrates this direct proportionality (see dash black line obtained with a linear regression applied to black-marked dots) and that a strong perturbation on $\alpha$ dramatically reduces the value of $r(T_{0},T)$ in the majority of cases as expected.

In conclusion, since consensus for the PI-ACE dynamics \eqref{eq:PIaverageconsensusestimator} is a necessary condition for the correct $\lambda_{2}^{L}$ estimation process performed by the DPIA, our proposed guarantees find a deep relevance in the secure design for such applications employing this kind of decentralized estimation algorithm.

\begin{figure}[t!]
	\centering
	\includegraphics[scale=0.115]{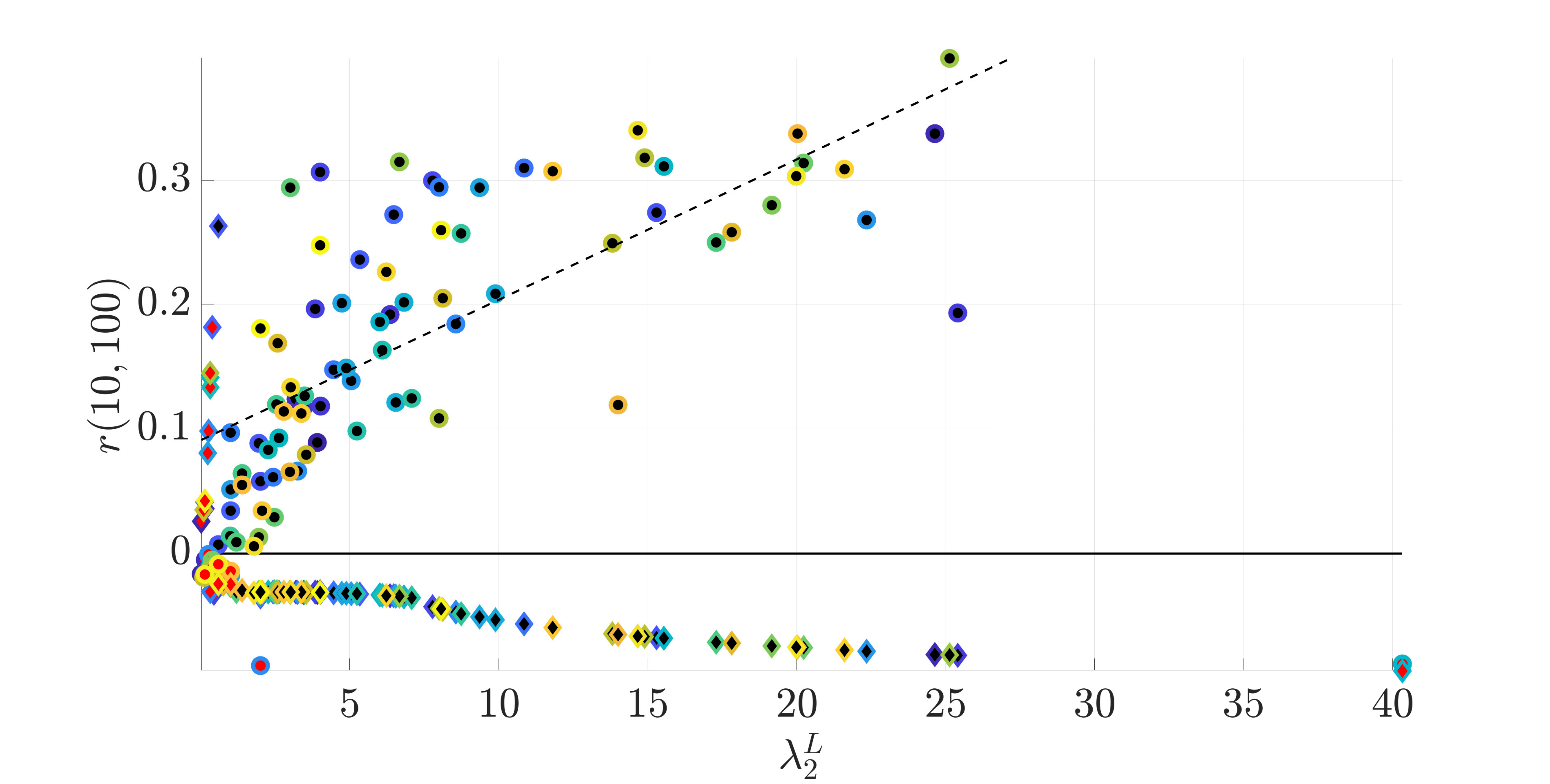}
	\caption{Computation of the convergence rate for several different random topologies, depicted via diverse colored markers. Dots and diamonds represent respectively the results for the nominal DPIA and the perturbed DPIA through $\delta_{\alpha}^{\theta} = -0.99 \alpha$ according to setup $S2$. Items marked in black are acceptable while those marked in red are not, as $r(10,100)\leq0$ for their associated nominal simulation.
	}
	\label{fig:lambda2vsr}
\end{figure}


\section{Extension to the discrete-time domain} \label{sec:dt_ext}
In this section, we propose an extension for the secure-by-consensus approach previously devised to the discrete-time domain. Within this framework, we let $t \in \mathbb{N}$ to indicate, without confusion, the discrete time instants and we assume the same setup proposed in the introductory part of Sec. \ref{sec:sbdconsensus} and through all Sec. \ref{ssec:objcod_infloc}. 

\subsection{Secure-by-design consensus in discrete time}
We consider and investigate a well-known discrete-time average consensus dynamics, namely that described by
\begin{equation}\label{eq:Consdt}
	\x(t+1) = \x(t) - \epsilon \L(\Gmc)\x(t)   = \F_{\epsilon}(\Gmc)\x(t),
\end{equation}
where $\epsilon$ is a common parameter shared among all agents and designed to belong to the interval $(0,2/\lambda_{n}^{L})$, see \cite{OlfatiSaberFaxMurray2007,FabrisMichielettoCenedese2022}. Constant $\epsilon$ is, indeed, selected in order to allow the state matrix $\F_{\epsilon}(\Gmc) = I_{N}- \epsilon \L(\Gmc)$ to be doubly stochastic with exactly $M$ eigenvalues equal to $1$ and all the remaining eigenvalues having modulus smaller than $1$ (\cite{OlfatiSaberFaxMurray2007,Chung1997}). Matrix $\F_{\epsilon}(\Gmc)$ can be further decomposed as $\F_{\epsilon}(\Gmc) = (F_{\epsilon}(\Gmc) \otimes I_{D})$, in which $F_{\epsilon}(\Gmc) = I_{n} - \epsilon L(\Gmc)$ is doubly stochastic and has eigenvalues $\lambda_{i}^{F_{\epsilon}} = 1-\epsilon \lambda_{i}^{L}$, for $i=1,\ldots,n$, ordered as $1=\lambda_{1}^{F_{\epsilon}} > \lambda_{2}^{F_{\epsilon}} \geq \cdots \geq \lambda_{n}^{F_{\epsilon}}$.
According to the characterization of the decoupling between objective coding and information localization in \eqref{eq:SBDC}, dynamics \eqref{eq:Consdt} can be rewritten as
\begin{equation}\label{eq:SBDCdt}
	\x(t+1) = \x(t) - \epsilon \H(\x(t)) \p(\theta),
\end{equation}
since it has been shown that $\H(\x) \p(\theta) = \L(\Gmc)\x$ in Sec. \ref{ssec:SBDCdyn} through Lem. \ref{lem:standardandsecurepersp}.

In the next paragraph, we will explore how this kind of discrete-time consensus protocol behaves whenever an encoded edge weight is perturbed by an attacker.

\subsection{Robustness to channel tampering in discrete time}
Adopting the same background and attack models introduced in Sec. \ref{sec:robust-analysis}, the $i$-th component, $i=1,\ldots,n$, of the perturbed dynamics associated to \eqref{eq:SBDCdt} is yielded by 
\begin{equation}\label{eq:perturbedSBDCdt}
	x_{i}(t+1) = x_{i}(t) - \epsilon {\textstyle\sum}_{j \in \Nmc_{i}} p_{ij}(\theta_{ij} + \delta^{\theta}_{ij}) h_{ij}(\x(t)),
\end{equation}
similarly to the altered description provided in \eqref{eq:perturbedSBDC}.
It is possible then to state the discrete-time version of Thm. \eqref{thm:maximum_perturb_codeword} for the perturbed protocol \eqref{eq:perturbedSBDCdt} as follows.
\begin{thm} \label{thm:maximum_perturb_codeword_dt}
	Assume that the characterization \textit{(i)-(iii)} in Subsec. \ref{sec:ressbdc} for objective decoding functions $p_i$ holds and recall $\Psi_{\Gmc}$ defined in Sec. \ref{sec:prelim_models}.
	Let an injection attack $\delta^{\theta} \in \bm{\Delta}^{\theta}$ affect a single edge $(u,v)\in \Emc$, i.e., $\delta^\theta_{ij} = 0$ for all $(i,j) \in \Emc \setminus \{(u,v)\}$ is satisfied, and define
	\begin{equation}
		\psi_{i}(\delta^\theta_{uv}) = w_{i}+K_{uv}|\delta^\theta_{uv}|, \quad i = u,v.
	\end{equation}
	Then the perturbed consensus protocol \eqref{eq:perturbedSBDCdt} reaches robust agreement for all $\delta^\theta_{uv}$ such that both \eqref{eq:sgrepatt} and
	\begin{equation}\label{eq:sgrepatt_dt}
		\phi_{\Gmc}(\delta^\theta_{uv}) := \max \{ \Psi_{\Gmc}, \psi_{u}(\delta^\theta_{uv}) , \psi_{v}(\delta^\theta_{uv}) \}  <  \epsilon^{-1}
	\end{equation}
	hold for any fixed $\epsilon$, independently from the values taken by any codeword $\theta \in \Theta$. 
\end{thm}
\begin{proof}
	To assess agreement for protocol \eqref{eq:perturbedSBDCdt} we first investigate the spectral properties of $F_{\epsilon} + \Delta^{F_{\epsilon}} = I_{n} - \epsilon (L + \Delta^{L}) = I_{n} - \epsilon E (W+\Delta^{W}) E^{\top}$, where quantity $\Delta^{F_{\epsilon}} = -\epsilon \Delta^{L} = -\epsilon E \Delta^{W} E^{\top}$ captures the uncertainty w.r.t. $F_{\epsilon}$ caused by a time-varying weight variation $\Delta^{W} = \delta_{uv}^{w} \mathbf{e}_z \mathbf{e}_z^{\top} $, with $z=(u,v)$. In order to ensure robust agreement in absence of objective coding, i.e. when $p_{ij}(\theta_{ij}) = \theta_{ij} = w_{ij}$ for all $(i,j) \in \Emc$ holds with no uncertainty, one imposes
	\begin{equation}\label{eq:stab_eig_condition}
		\left|\lambda_{i}^{F_{\epsilon} + \Delta^{F_{\epsilon}}} \right| = \left| 1-\epsilon \lambda_{i}^{L+\Delta^{L}} \right|   < 1, \quad i = 2,\ldots,n.
	\end{equation}
	To satisfy condition \eqref{eq:stab_eig_condition} it is sufficient to ensure both
	\begin{equation}\label{eq:eig1Fcond}
		\lambda_{1}^{L + \Delta^{L}} > 0,
	\end{equation}
	\begin{equation}\label{eq:eignFcond}
		\lambda_{n}^{L + \Delta^{L}}/2 < \epsilon^{-1}.
	\end{equation}
	Inequality \eqref{eq:eig1Fcond} is guaranteed to hold if \eqref{eq:fundamineqres} holds\footnote{Under a perturbation on a single edge weight, linear agreement to a unique value  can be reached if and only if \eqref{eq:fundamineqres} is satisfied.} through Lem. \ref{thm:effective_resistance}. Whereas, condition \eqref{eq:eignFcond} foists a further requirement to achieve stability w.r.t. to the continuous-time case.
	
	By resorting to the Gershgorin circle theorem \cite{Bell1965}, it is possible to find an upper bound for $\lambda_{n}^{L + \Delta^{L}}$ and ensure \eqref{eq:eignFcond} as follows. If $\delta_{uv}^{w} = 0$, i.e. considering the nominal system \eqref{eq:Consdt}, then $\lambda_{n}^{L + \Delta^{L}} \leq 2\Psi_{\Gmc}$. Otherwise, if $\delta_{uv}^{w} \neq 0$, it is possible that the following couple of inequalities may also be useful to find an upper bound: $\lambda_{n}^{L + \Delta^{L}} \leq 2 (w_{i} + |\delta_{uv}^{w}|)$, with $i \in \{u,v\}$. To summarize, setting $\bar{\phi}_{\Gmc}(\delta_{uv}^{w}) :=  \max\{\Psi_{\Gmc}, (w_{u} + |\delta_{uv}^{w}|), (w_{v} + |\delta_{uv}^{w}|) \}$ the following upper bound can be provided for all $\delta_{uv}^{w} \in \mathbb{R}$:
	\begin{equation}\label{eq:ubwln}
		\lambda_{n}^{L + \Delta^{L}}/2 \leq  \bar{\phi}_{\Gmc}(\delta_{uv}^{w}).
	\end{equation}
	
	Now, to guarantee the robust agreement in presence of objective coding, we recall inequality \eqref{eq:concavity} and the fact that $|\delta_{uv}^{w}| \leq K_{ij}|\delta^{\theta}_{ij}|$. It is, thus, straightforward to observe that 
	$\bar{\phi}_{\Gmc}(\delta_{uv}^{w}) \leq \phi_{\Gmc}(\delta^\theta_{uv}) =  \max\{\Psi_{\Gmc}, \psi_{u}(\delta^\theta_{uv}), \psi_{v}(\delta^\theta_{uv}) \} $. Therefore, thanks to \eqref{eq:ubwln}, the imposition of \eqref{eq:sgrepatt_dt} is sufficient to satisfy \eqref{eq:eignFcond}.
\end{proof}

\begin{rmk}
	It is crucial to observe that inequality \eqref{eq:sgrepatt_dt} is conservative as the topology of $\Gmc$ varies, even for decoding functions $p_{ij}$ linear in their argument. However, this is not the case if: (a) the latter decryption for $\theta$ is chosen (this, indeed, allows equality $\bar{\phi}_{\Gmc}(\delta_{uv}^{w}) = \phi_{\Gmc}(\delta_{uv}^{\theta})$ to be attained); (b) the topology under consideration satisfies $\Psi_{\Gmc} = \lambda_{n}^{L}/2$, namely if $\Psi_{\Gmc}$ represents the infimum for the values taken by $\epsilon^{-1}$ (we recall that $\epsilon \in (0,2/\lambda_{n}^{L})$). An example for such topologies is the class of uniformly weighted regular bipartite networks. 
	Indeed, these networks are characterized by $\Psi_{\Gmc} = wd = \lambda_{n}^{L}/2$ (see \cite{Chung1997}). 
\end{rmk}

In addition to this, the main result obtained in Thm. \ref{thm:maximum_perturb_codeword_dt} can be further simplified by means of the following corollary.
\begin{cor} \label{cor:maximum_perturb_codeword_dt}
	Under all the assumptions adopted in Thm. \ref{thm:maximum_perturb_codeword_dt} and setting $\epsilon < \Psi_{\Gmc}^{-1}$, the perturbed consensus protocol \eqref{eq:perturbedSBDCdt} reaches robust agreement for all $\delta^\theta_{uv}$ such that
	\begin{equation}\label{eq:sgrepatt_dt_simple}
		|\delta_{uv}^{\theta}| < \xi_{uv}^{\theta} :=  K_{uv}^{-1} \min \{ \Rmc_{uv}^{-1}(\Gmc) , (\epsilon^{-1}-\Psi_{\Gmc}) \}
	\end{equation}
	independently from the values taken by any codeword $\theta \in \Theta$. In particular, condition \eqref{eq:sgrepatt} needs to be fulfilled solely to guarantee consensus if $\epsilon$ is selected as follows:
	\begin{equation}\label{eq:epsguar_rs}
		\epsilon \leq \epsilon^{\star}_{uv} := (\Psi_{\Gmc} + \Rmc_{uv}^{-1}(\Gmc))^{-1}.
	\end{equation}
\end{cor}
\begin{proof}
	Relation in \eqref{eq:sgrepatt_dt_simple} is the combined result of guarantee in \eqref{eq:sgrepatt} and that one obtainable by imposing $\Psi_{\Gmc} + K_{uv}|\delta_{uv}^{\theta}| < \epsilon^{-1}$ to satisfy \eqref{eq:sgrepatt_dt}, since $\phi_{\Gmc}(\delta^\theta_{uv})$ can be upper bounded as $\phi_{\Gmc}(\delta^\theta_{uv}) \leq \Psi_{\Gmc} + K_{uv}|\delta_{uv}^{\theta}|$. On the other hand, relation \eqref{eq:epsguar_rs} is derived by enforcing $\Rmc_{uv}^{-1}(\Gmc) \leq \epsilon^{-1}-\Psi_{\Gmc}$ to minimize $\xi_{uv}^{\theta}$ and obtain $\xi_{uv}^{\theta} = \rho_{uv}^{\theta}$, as, in general, one has $\xi_{uv}^{\theta} \leq \rho_{uv}^{\theta}$.
\end{proof}

Cor. \ref{cor:maximum_perturb_codeword_dt} highlights the fact that, in discrete time, robustness margin $\xi_{uv}^{\theta}$ is not only determined by quantity $\rho^{\theta}_{uv}= (K_{uv}\Rmc_{uv}(\Gmc))^{-1}$ but also strongly depends on the inversely proportional relationship between $\epsilon$ and $\Psi_{\Gmc}$. The smaller $ \Psi_{\Gmc}$ w.r.t. $\epsilon^{-1} $ the better robustness is achieved, up to the lower limit dictated by $\Rmc_{uv}^{-1}(\Gmc)$. Indeed, margins $\xi_{uv}^{\theta}$ and $\rho_{uv}^{\theta} $ coincide for $\epsilon \leq \epsilon^{\star}_{uv}$, namely $\xi_{uv}^{\theta}$ is minimized, as $\xi_{uv}^{\theta} \leq \rho_{uv}^{\theta}$ holds. This also suggests that discrete-time robust agreement may be harder to be reached w.r.t. the continuous-time case. 
Finally, from Cor. \ref{cor:maximum_perturb_codeword_dt} it can be easily noticed that 
\begin{equation}\label{eq:epschoice}
	\epsilon \leq \epsilon^{\star} := \underset{(i,j) \in \Emc}{\min} \epsilon_{ij}^{\star} = \left(\Psi_{\Gmc} + \underset{(i,j) \in \Emc}{\max} \Rmc_{ij}^{-1}(\Gmc)\right)^{-1}
\end{equation}
is a sufficient choice to provide the exact robustness guarantees as in the continuous-time framework, \textit{regardless the edge in $\Gmc$ being under attack}. Hence, parameter $\epsilon$ can be set ahead consensus protocol starts, according to \eqref{eq:epschoice} and without the full knowledge of each encrypted edge weight being sent by the network manager.



\begin{figure*}[t!] 
	\centering
	\subfigure[]{\includegraphics[height=0.18\textwidth, trim={0.1cm 0cm 0.1cm 0cm},clip ]{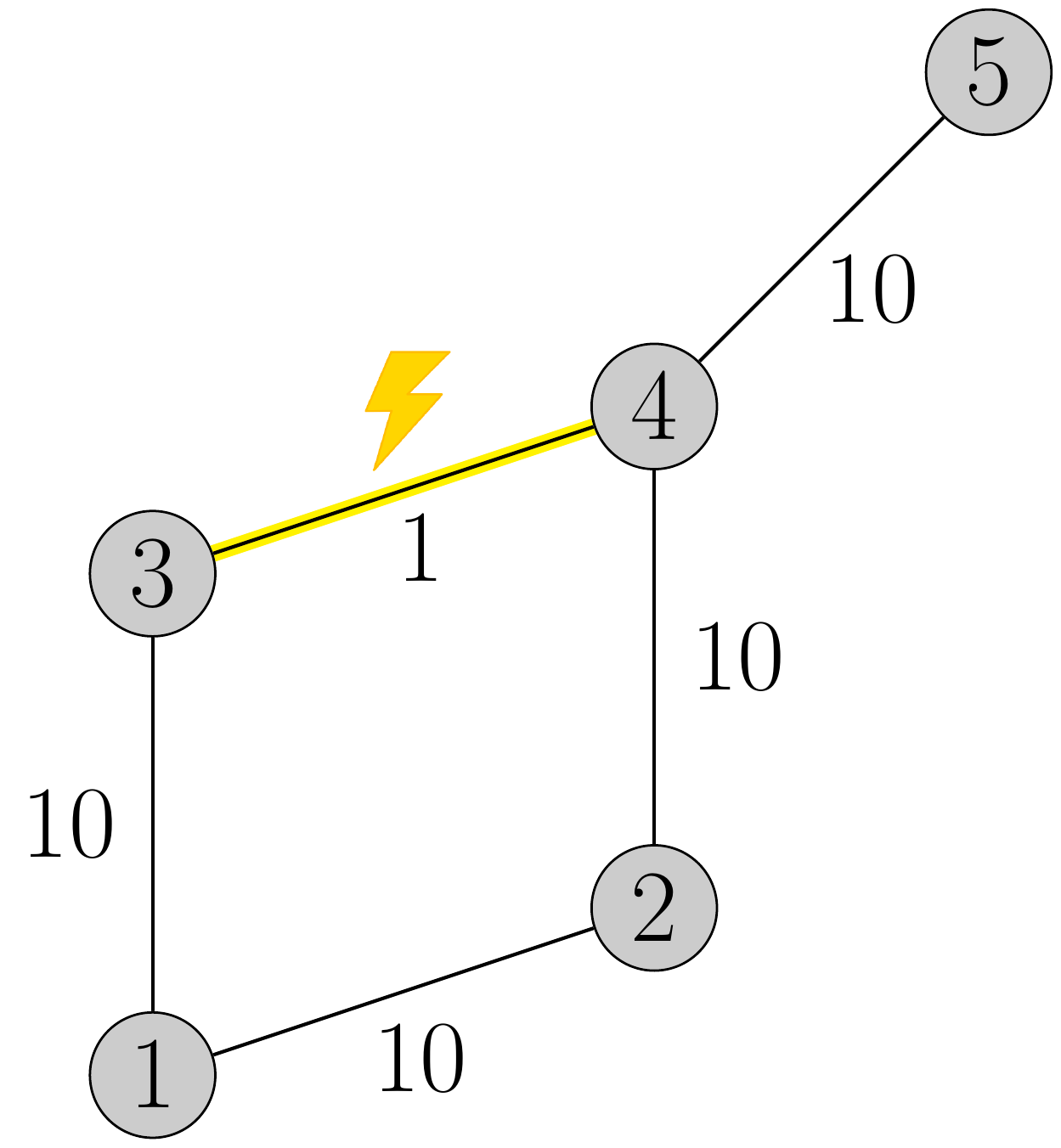}\label{fig:graph}}\hspace{1cm}
	\subfigure[]{\includegraphics[height=0.18\textwidth, trim={5.8cm 0cm 10cm 1cm},clip]{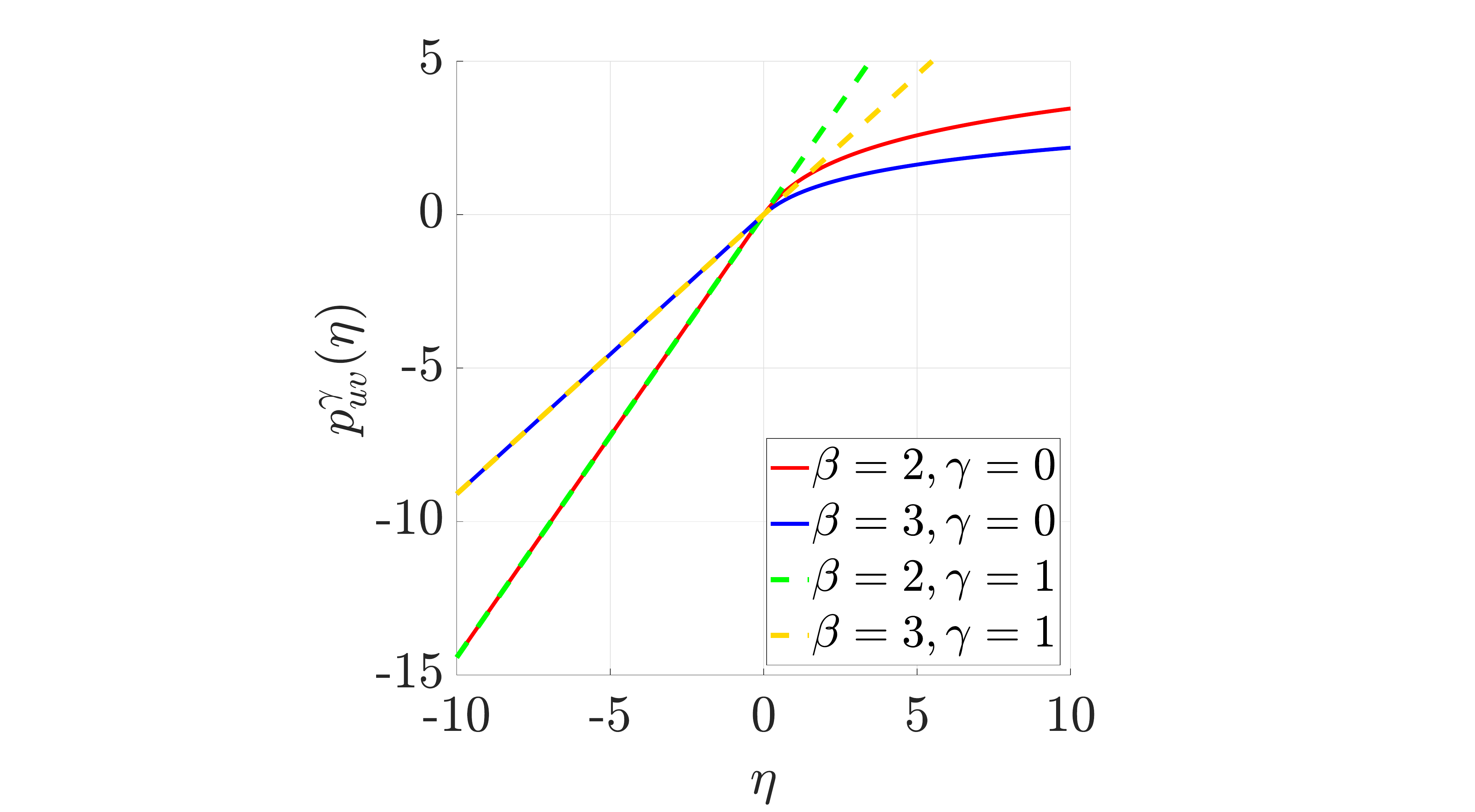}\label{fig:p_rs_zest}}\hspace{1cm}
	\subfigure[$p_{uv}^{0}$, $\delta_{uv}^{\theta} = -4.7$]{\includegraphics[height=0.18\textwidth, trim={9cm 0cm 10cm 1cm},clip]{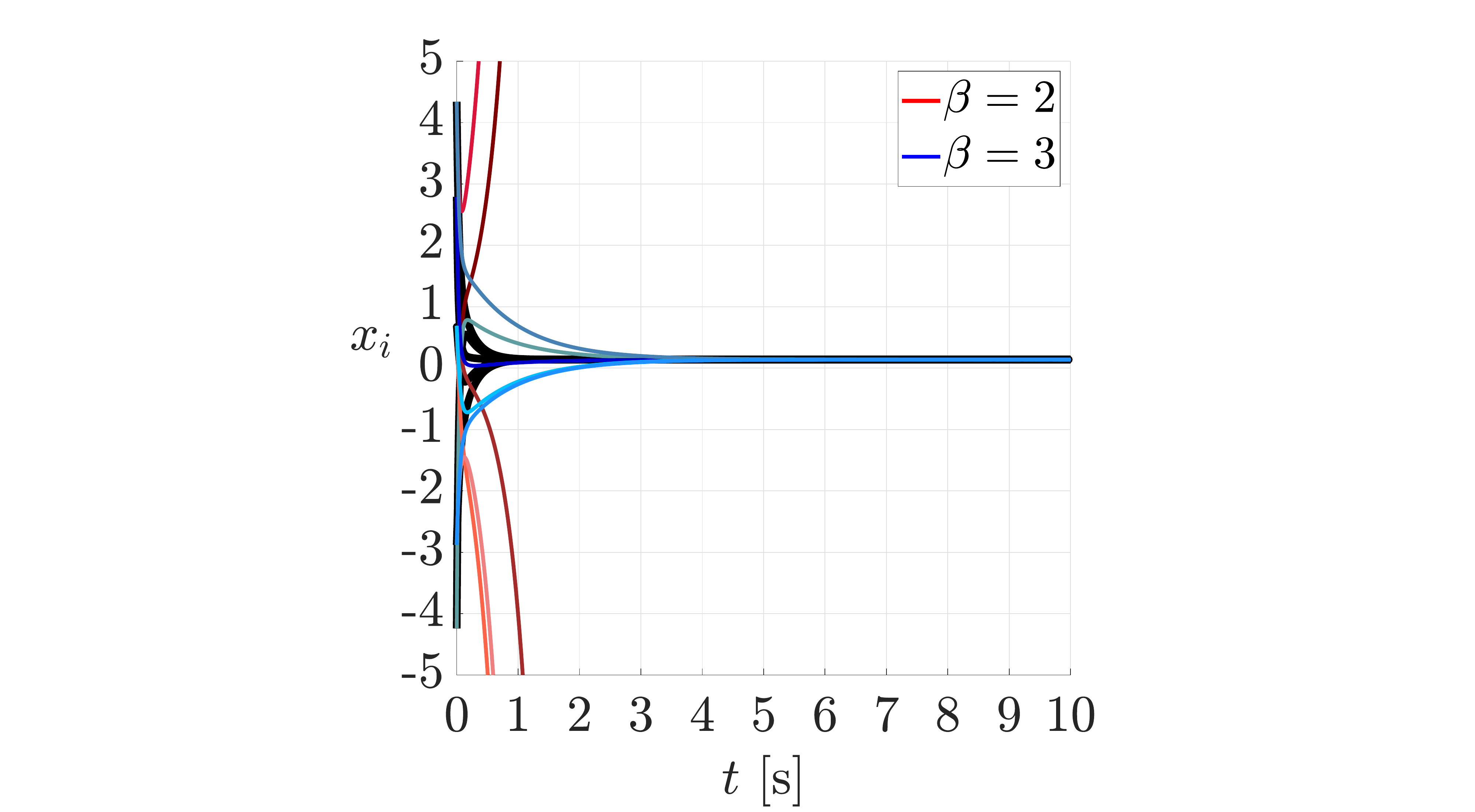}\label{fig:robust_cons_zest}}\hspace{1cm}
	\subfigure[$p_{uv}^{1}$, $\delta_{uv}^{\theta} \!\simeq\! -3.0036$]{\includegraphics[height=0.18\textwidth, trim={9cm 0cm 10cm 1cm},clip]{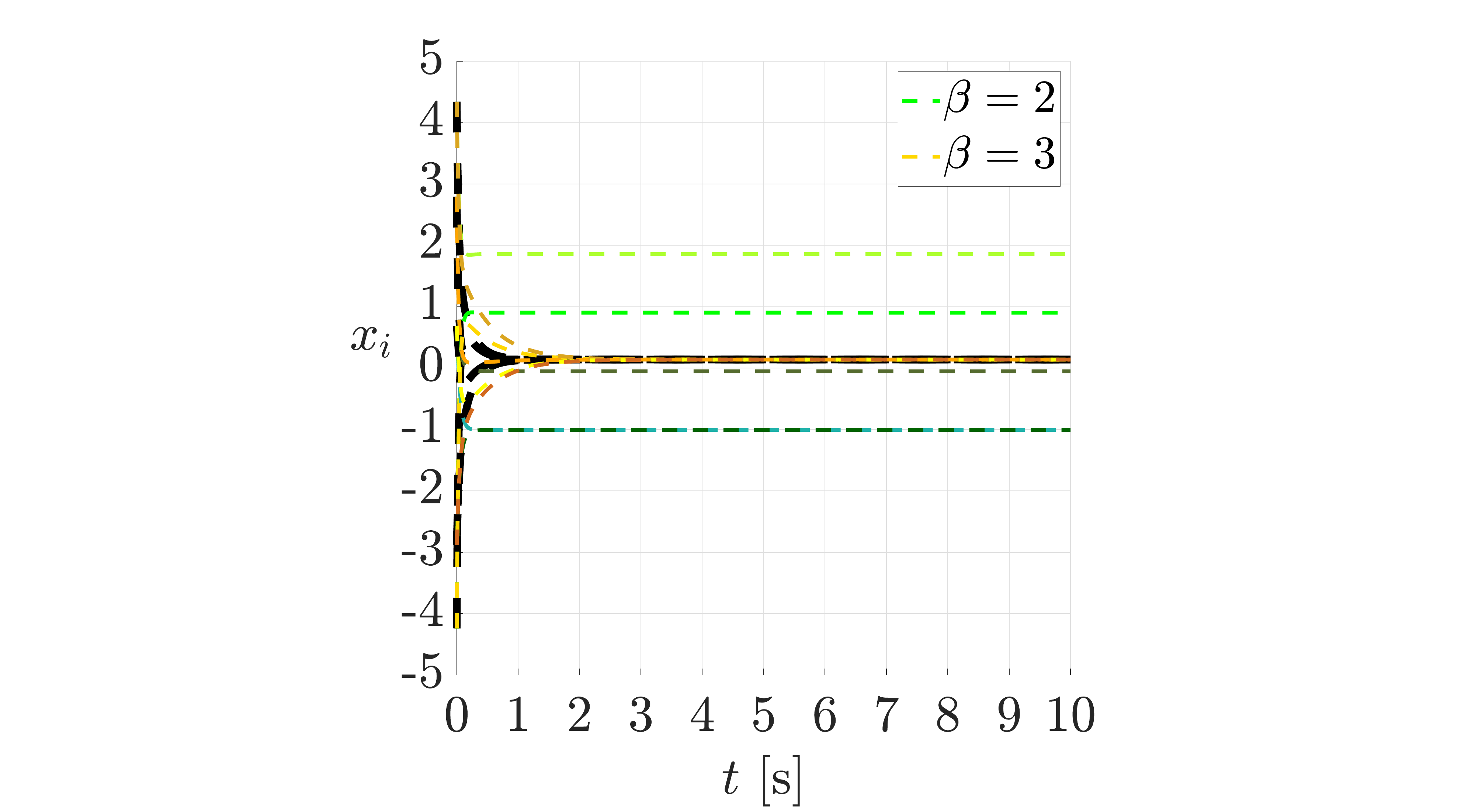}\label{fig:clust_cons_zest}}
	\caption{(a) Considered network topology and attack on edge $(u,v)=(3,4)$; (b) Decoding function in \eqref{eq:p_rs}, Lipschitz constants $K^{(\beta)}_{uv} = 1/\ln(\beta)$, $\beta = 2,3$, are highlighted (dashed lines); (c-d) Agent dynamics as objective coding and perturbation vary.}
	\label{fig:simulations}
\end{figure*}


\section{Numerical simulations} \label{sec:simulations}
Few numerical simulations are here provided to validate and motivate the theoretical results debated so far.

\subsection{Continuous-time example}
We now briefly report on a numerical simulation illustrating the main results of this work, within continuous-time framework presented in Secs. \ref{sec:sbdconsensus}-\ref{sec:robust-analysis}. Fig. \ref{fig:graph} shows the network topology analyzed. States $x_{i}$, with $i = 1,\ldots,n$, are assumed to be in $\Rset{}$, namely $D=1$. We suppose that a constant attack $\delta_{uv}^{\theta}$ strikes subcodeword $\theta_{uv}$ corresponding to the edge with the lowest weight,\footnote{In other words, the attacker attempts to cut down the link with highest network resistance.} i.e., $(u,v)=(3,4)$. The decoding functions for this edge, depicted in Fig. \ref{fig:p_rs_zest}, are chosen as
\begin{equation}\label{eq:p_rs}
	p^{0}_{uv}(\eta) = \begin{cases}
			\log_{\beta}(1+\eta),   & \eta \geq 0; \\
			\eta/\ln(\beta), & \eta < 0 ;
	\end{cases} \quad p^{1}_{uv}(\eta) = \dfrac{\eta}{\ln(\beta)}; 
\end{equation}
and are designed to return $w_{uv}=1$ for the  expected codeword input $\theta$ (i.e., $p_{uv}^\gamma (\theta) = w_{uv}$ for $\gamma=0,1$).
Moreover, in this setup, we adopt decoding functions $p_{ij}$ defined over the entire real set for sake of simplicity. Further generalizations may be implemented, as already suggested, by accounting for perturbed subcodewords $ (\theta_{ij}+\delta^{\theta}_{ij})$ falling outside the decoding function domains $\Theta_{ij}$ and declaring them invalid. Once received, these can then be used as alert to signal a certain ongoing threat. 

According to \eqref{eq:sgrepatt}, the maximum allowed perturbation in magnitude is yielded by $\rho_{uv}^{\theta} \simeq 3.0036$, for $\beta = 2$, and $\rho_{uv}^{\theta} \simeq 4.7607$, for $\beta =3$. In Fig. \ref{fig:robust_cons_zest}, it is possible to see that agreement takes place -- by virtue of Thm. \ref{thm:maximum_perturb_codeword} -- only for $\beta = 3$ and $p_{uv}^{0}$, if $\delta_{uv}^{\theta}=-4.7$. Here, black curves denote free-attack consensus trajectories ($\delta_{uv}^{\theta}=0$). It is worth to note that this attack leads to a negative perturbed weight on edge $(u,v)$ for both $\beta = 2,3$; indeed, to obtain $p_{uv}^{0}(\theta_{uv}) = w_{uv} = 1$, it is required for the network manager to send $\theta_{uv} = \beta-1$, implying that $p^{0}_{uv}(\theta_{uv}+\delta_{uv}^{\theta}) < p^{0}_{uv}(\beta-3) = (\beta-3)/\ln(\beta) \leq 0$.
The latter simulation also highlights the tradeoff in Prop. \ref{fact:tradeoff} between encryption capability of $p_{uv}^{0}$ and $p_{uv}^{1}$, in terms of Lipschitz constant $K_{uv}^{(\beta)}$, and the robustness achieved w.r.t. edge $(u,v)$. Indeed, on one hand, it is immediate to realize that $K_{uv}^{(2)} = 1/\ln(2) > K_{uv}^{(3)} = 1/\ln(3)$ implies that $p_{uv}^{0}$, $\beta = 2$, reaches a wider range of values compared to $p_{uv}^{0}$, $\beta = 3$ -- given the same interval $U_{uv}^{\theta}$ -- thus leading to higher encryption performances. On the other hand, it is worth to notice that, in case of $\delta_{uv}^{\theta} = -4.7$, for $\beta = 2$ the network does not even attain consensus but the opposite occurs if $\beta = 3$.
Furthermore, for $p_{uv}^{1}$, Prop. \ref{prop:tradeoff} applies and the effects of tradeoff in Prop. \ref{fact:tradeoff} become strict (see Fig. \ref{fig:clust_cons_zest}; still, black curves denote free-attack consensus trajectories). Indeed, for $\delta_{uv}^{\theta}=-\rho_{uv}^{\theta}$, the well-known clustered consensus phenomenon arises for $\beta = 2$, since the corresponding stability margin is nullified.
Lastly, it is also worth to observe that, for both $p_{uv}^{0}$ and $p_{uv}^{1}$, agent trajectories for $\beta = 3$ have faster convergence rate w.r.t. those for $\beta = 2$, justifying the possibility for a diverse edge weight choice by the network manager.

\subsection{Discrete-time example on opinion dynamics}
In this last paragraph, we provide a numerical example based on the opinion dynamics work proposed in \cite{MorarescuGirard2011}. We consider the uniformly weighted opinion network $\Gmc_{\alpha} = (\Vmc,\Emc,\{\alpha\}_{k=1}^{m})$, with $\alpha \in \mathcal{Q}_{\alpha} = (0,1/2)$, such that $(\Vmc,\Emc)$ describes the same topology in Fig. \ref{fig:graph}. Assuming $t \in \mathbb{N}$, let us also define the time-varying $i$-th \textit{opinion neighborhood} as
\begin{equation*}
	\Nmc_{i}(t) = \{j \in \Vmc \;|\; ((i,j)\in \Emc) \wedge (|x_{i}(t)-x_{j}(t)| \leq \Gamma \upsilon^{t})\},
\end{equation*}
where $\Gamma > 0$ and $\upsilon \in (0,1)$ are given.
Each agent $i \in \{1,\ldots,n\}$ in the opinion network is then assigned with the perturbed discrete-time opinion dynamics
\begin{align}\label{eq:opdyn_sys}
	& x_{i}(t+1) =  \\ 
	& \begin{cases}
		x_{i}(t),&\text{ if } \Nmc_{i}(t) = \emptyset ;\\
		x_{i}(t) 
		- \frac{1}{|\Nmc_{i}(t)|} \sum\limits_{j \in \Nmc_{i}(t)} w^{\delta^{\theta}}_{ij} (x_{i}(t)-x_{j}(t)), &\text{ otherwise}; \nonumber\\
	\end{cases}
\end{align}
where $x_{i}(t) \in \mathbb{R}$ and each $w^{\delta^{\theta}}_{ij} = p_{ij}(\theta_{ij} + \delta_{ij}^{\theta})$ represents the perturbed decoded value, with $p_{ij}(\theta_{ij}) = \alpha/\ln(2)$, $\forall (i,j) \in \Emc$. Despite \eqref{eq:opdyn_sys} does not possess the exact same form of protocol \eqref{eq:perturbedSBDCdt}, it is possible to provide a brief analysis of its behavior when certain setups are fixed. Indeed, term $\epsilon(t) := |\Nmc_{i}(t)|^{-1}$ can be seen as a time-varying version of $\epsilon$, upper bounded by $\overline{\epsilon} = 1$. Since the maximum attainable node degree $d_{M} = \max_{i\in \{1,\ldots,n\}} |\Nmc_{i}(t)|$ in $\Gmc_{\alpha}$ over time is $d_{M} = 3$, one has $\Psi_{\Gmc_{\alpha}} = d_{M} \alpha = 3 \alpha $ and, according to \eqref{eq:sgrepatt_dt}, inequality $\Psi_{\Gmc_{\alpha}} < \epsilon(t)^{-1}$ can be reduced to $\Psi_{\Gmc_{\alpha}} < \overline{\epsilon}^{-1}$, yielding the design constraint $\alpha \in  (0,1/3) \subset \mathcal{Q}_{\alpha}$. Assuming, once again, that edge $(3,4)$ is subject to an attack $\delta_{34}^{\theta}$, parameter $\alpha$ can be selected to maximize the r.h.s. of guarantee \eqref{eq:sgrepatt_dt_simple}, by imposing $1-3\alpha = 4\alpha/3$ and obtaining $\alpha = 3/13 \in (0,1/3)$.

Fig. \ref{fig:opdyn} shows the trajectories of opinion dynamics \eqref{eq:opdyn_sys} once initialized with $\Gamma = 10$, $\upsilon = 1-0.2 \alpha = 0.9538$ and $x(0) = \begin{bmatrix}
	-3.2 & -1 & 3.3 & 3 & -4.3
\end{bmatrix}^{\top}$. Remarkably, within this setup, guarantee \eqref{eq:sgrepatt_dt_simple} is not conservative w.r.t. \eqref{eq:sgrepatt_dt}, since each deconding function has the same Lipschitz constant and edge $(3,4)$ is incident to node $4$, which has the highest degree $d_{M}$. This evidence and the fact that the topology under analysis is bipartite and uniformly weighted imply that inequality \eqref{eq:sgrepatt_dt_simple} may yield a sharp guarantee for the robust consensus through certain choices of $\Gamma$ and $\upsilon$. Indeed, this is the case for simulations in Fig. \ref{fig:opdyn}, in which it is possible to appreciate that for $\delta_{34}^{\theta} = 0$ the system nominally converges to consensus (green lines), forming one community, i.e. $\Vmc$; while for $\delta_{34}^{\theta} = -\xi_{uv}^{\theta} = -0.21328$ clustered consensus phenomena arise for $t \leq 70~s$ (red lines). Afterwards, for $t > 70 ~s$, the five separated communities $\{1\}$, $\{2\}$, $\{3\}$, $\{4\}$, $\{5\}$ merge because of the nonlinearities in the opinion dynamics \eqref{eq:opdyn_sys}. Finally, it is also worth to observe that, if $\delta_{34}^{\theta} = -6 \xi_{uv}^{\theta} = -1.2797$, the attack asymptotically prevents consensus to be achieved (blue lines), causing the permanent split into a couple of diverse communities, i.e. those constituted by nodes $\{1,2,4,5\}$ and $\{3\}$, as information exchange stops flowing through edges $(1,3)$ and $(3,4)$. In other words, the latter attack manages to isolate node $3$ from the original opinion network, leading to a completely different scenario w.r.t. to the nominal, as $t \rightarrow \infty$.

\begin{figure}[t!]
	\centering
	\includegraphics[scale=0.105]{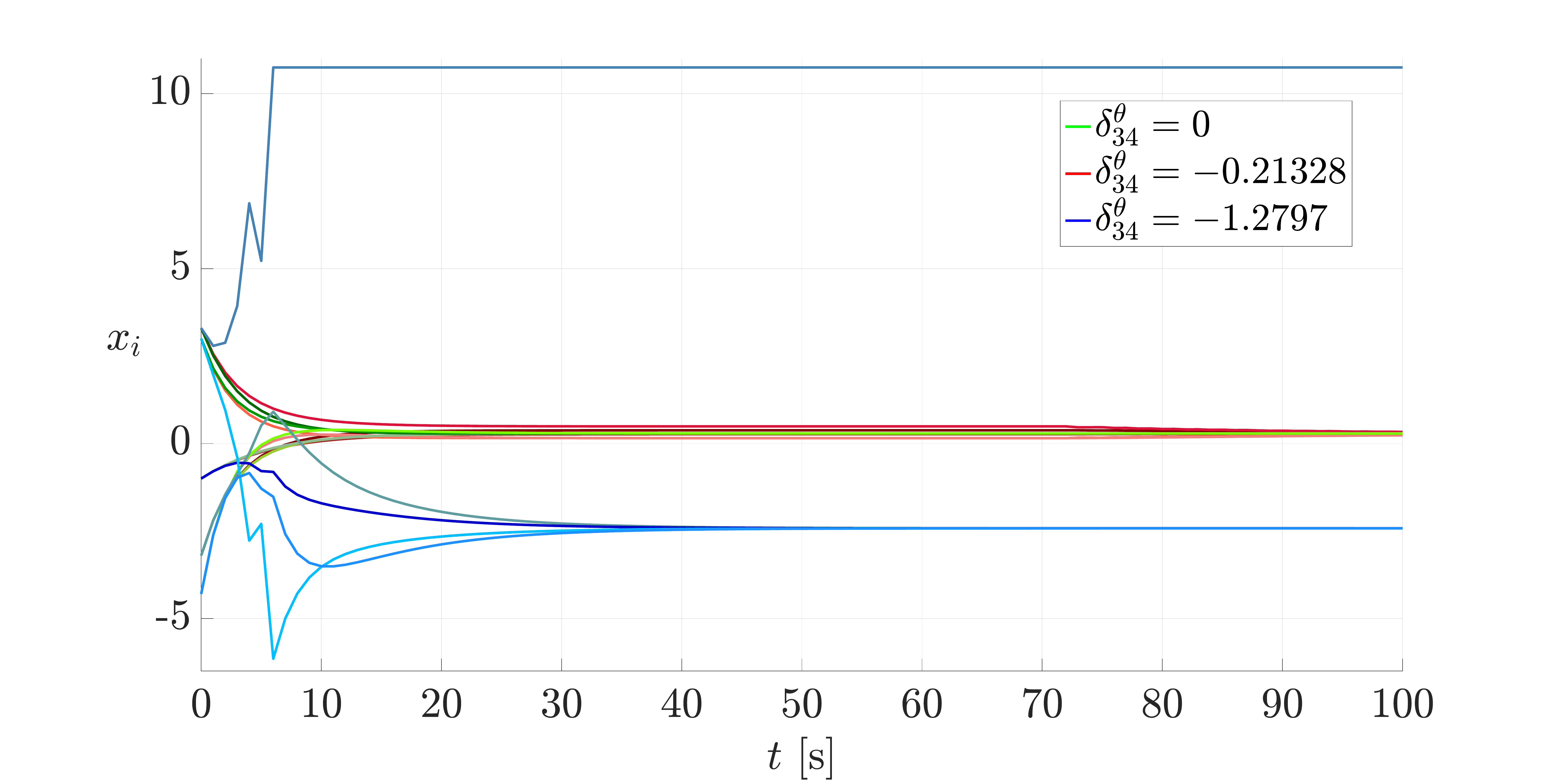}
	\caption{Results obtained simulating system \eqref{eq:opdyn_sys} subject to different perturbations on edge $(3,4)$ in $\Gmc_{\alpha}$, with $\alpha = 3/13$.
	}
	\label{fig:opdyn}
\end{figure}

\section{Conclusions and future directions} \label{sec:conclusions}
This paper devises novel methods to secure consensus networks both in the continuous and discrete time domains, 
providing small-gain-theorem-based stability guarantees and a deep insight on a tradeoff between information hiding and robust stability.
Future works will involve extensions towards other multiagent protocols, such as distance-based formation control, and leader-follower or multi-attack scenarios. The security and estimation accuracy improvement of filtering algorithms within multisensor networks is also envisaged. 

\appendix \label{appendix}
\begin{proof}[Proof of Prop. \ref{Prop:eigM}]
	From the eigenvalue equation $M \omega = \lambda \omega$ in the unknowns $\lambda \in \mathbb{C}$ and $\omega = \begin{bmatrix}
		\omega_{1}^{\top} & \omega_{2}^{\top}
	\end{bmatrix}^{\top}$, with $\omega_{1},\omega_{2} \in \mathbb{C}^{n}$, one obtains the system of equations
	\begin{equation}\label{eq:syseqeigM}
		\begin{cases}
			(K_{P}L+\alpha I_{n}) \omega_{1} - K_{I} L \omega_{2} = \lambda \omega_{1} \\
			K_{I} L\omega_{1} = \lambda \omega_{2}
		\end{cases}.
	\end{equation}
	Note that if $\lambda = 0$, $\omega_{1} = \zerovec{n}$ and $\omega_{2} \in \left\langle \ones_{n} \right\rangle $ then \eqref{eq:syseqeigM} holds true.
	The second equation in \eqref{eq:syseqeigM} suggests that relation
	\begin{equation}\label{eq:eigrelationM}
		(s K_{I} \mu , \omega_{1}) = \left(\lambda, s\omega_{2}\right), \text{for some } s \in \mathbb{C},
	\end{equation}
	characterizes all the eigenpairs $(\mu,  \omega_{*} ) \in (\mathbb{R}_{\geq0},\left\langle \omega_{2} \right\rangle)$ associated to the Laplacian $L$, except for some of the configurations described by $\mu = 0$ or $\omega_{2} = \zerovec{n}$. 
	Substituting \eqref{eq:eigrelationM} into the first equation of \eqref{eq:syseqeigM} multiplied by $s$ at both sides one obtains the second order algebraic equation in the unknown $s$,
	\begin{equation}\label{eq:2oeqeigs}
		(K_{I} \mu s^{2} -(\alpha + K_{P}\mu) s + K_{I} \mu) \omega_{1} = \zerovec{n}. 
	\end{equation}
	If $\omega_{1} = \zerovec{n}$, the only acceptable value of $s$ complying with relation \eqref{eq:eigrelationM}, as $\omega_{2} \neq \zerovec{n}$ in general, 
	is given by $s^{\star} = 0$ with single algebraic multiplicity, since this result is derived from $\omega_{1} = s \omega_{2}$. 
	Otherwise, if $\omega_{1} \neq \zerovec{n}$ and $ \mu \neq 0$, the solutions are now given by $s=s_{\pm}$, where
	\begin{equation}\label{eq:allsinfuncofmu}
		s_{\pm} = \dfrac{\alpha+K_{P}\mu \pm \sqrt{(\alpha+K_{P}\mu)^{2}-4(K_{I}\mu)^{2}}}{2K_{I}\mu}.
	\end{equation}
	Also, if $\mu = 0$, a trivial solution is, again, $s^{\star} = 0$ with single algebraic multiplicity, by solving $\alpha s = 0$.
	Finally, substituting \eqref{eq:allsinfuncofmu} into relation $\lambda = sK_{I}\mu$, it follows that the eigenvalues of $M$ are given by \eqref{eq:eigsofMinfuncofeigsofL}-\eqref{eq:eigsofMinfuncofeigsofL_particular}.
	In particular, the evaluation at $i=1$ for both $j=1,2$ in \eqref{eq:eigsofMinfuncofeigsofL} requires $\lambda_{1}^{L} = 0$, i.e. involving case $\mu = 0$. The arithmetic extension of \eqref{eq:eigsofMinfuncofeigsofL}-\eqref{eq:eigsofMinfuncofeigsofL_particular} to this peculiar instance is obtained as follows.
	Case $i=1$ and $j=1$ is trivial, as $\lambda_{1}^{M} = 0$ occurs for $s^{\star} = 0$ in \eqref{eq:2oeqeigs}, if $\mu = 0$ or $\omega_{1} = \zerovec{n}$. 
	Case $i=1$ and $j=2$, corresponding to $\lambda_{2}^{M} = \alpha $, can be proven by exclusion (it is the only eigenvalue that relation \eqref{eq:eigrelationM} cannot describe) 
	and inspection. Indeed, 
	by selecting $\lambda = \lambda_{2}^{M}$, $\omega_{1} \in \left\langle \ones_{n} \right\rangle $, $\omega_{2} = \zerovec{n}  $ so that system \eqref{eq:syseqeigM} holds true. 
	
	The final part of the statement in the proposition is proven as follows. Firstly, recall that $\lambda_{1}^{M} = 0$ and $\lambda_{2}^{M} =  \alpha> 0$. Secondly, relation $\Re[\lambda_{l}^{M}] > 0$ for $l = 3,\ldots,2n$ is a consequence of the fact that if $\sigma_{i}$ is purely imaginary then the thesis is guaranteed to hold, as $\varphi_{i} > 0$, $\forall i = 2,\ldots,n$; otherwise, solving $\Re[\lambda_{l}^{M}] > 0$ for any $l \in \{ 3,\ldots,2n \}$, whenever $\sigma_{i} \in \mathbb{R}$, leads to the tautology $\lambda_{i}^{L} > 0$ for the corresponding $i \in \{ 2,\ldots,n \}$.
\end{proof}


%

\ifCLASSOPTIONcaptionsoff
  \newpage
\fi

\bibliographystyle{IEEEtran}
\bibliography{biblio}
\vspace{-0.7cm}
\begin{IEEEbiography}[{\includegraphics[width=1in,height=1.25in,trim={2.5cm 0 2.5cm 0},clip,keepaspectratio]{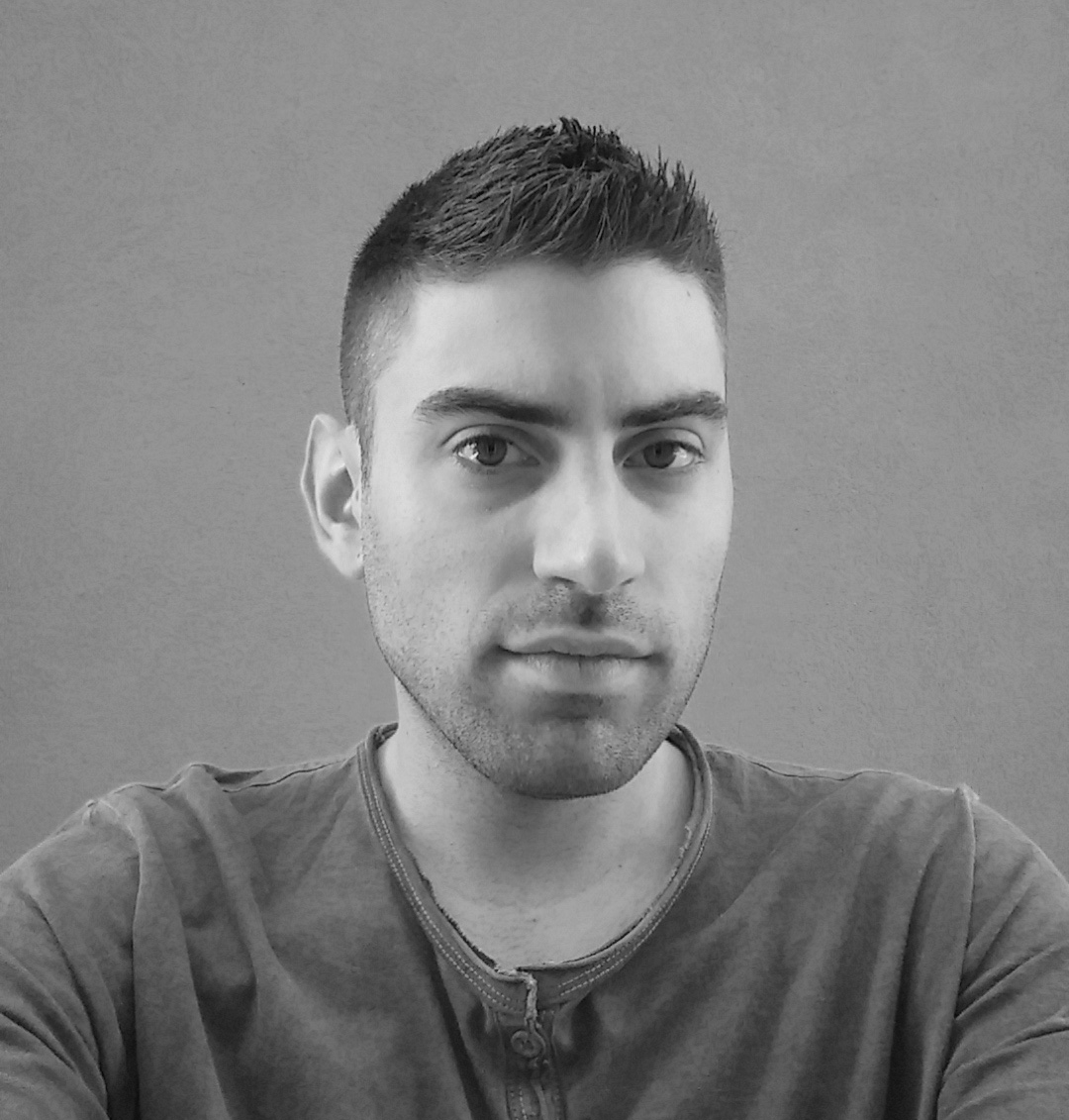}}]{Marco Fabris} 
	received the Laurea (M.Sc.) degree (summa cum laude) in automation engineering and the Ph.D. degree in information science and technology from the University of Padua, Padua, Italy, in 2016 and 2020, respectively.\\
	In 2018, he spent six months with the University of Colorado Boulder, Boulder, CO, USA, as a Visiting Scholar, focusing on distance-based formation control. He is currently working as a Postdoctoral Fellow with the Technion--Israel Institute of Technology, Haifa, Israel. His current research interests also involve graph-based consensus theory and optimal decentralized control and estimation for networked systems.	
\end{IEEEbiography}
\vspace{15cm}
\begin{IEEEbiography}[{\includegraphics[width=1in,height=1.25in,trim={2.5cm 0 2.5cm 0},clip,keepaspectratio]{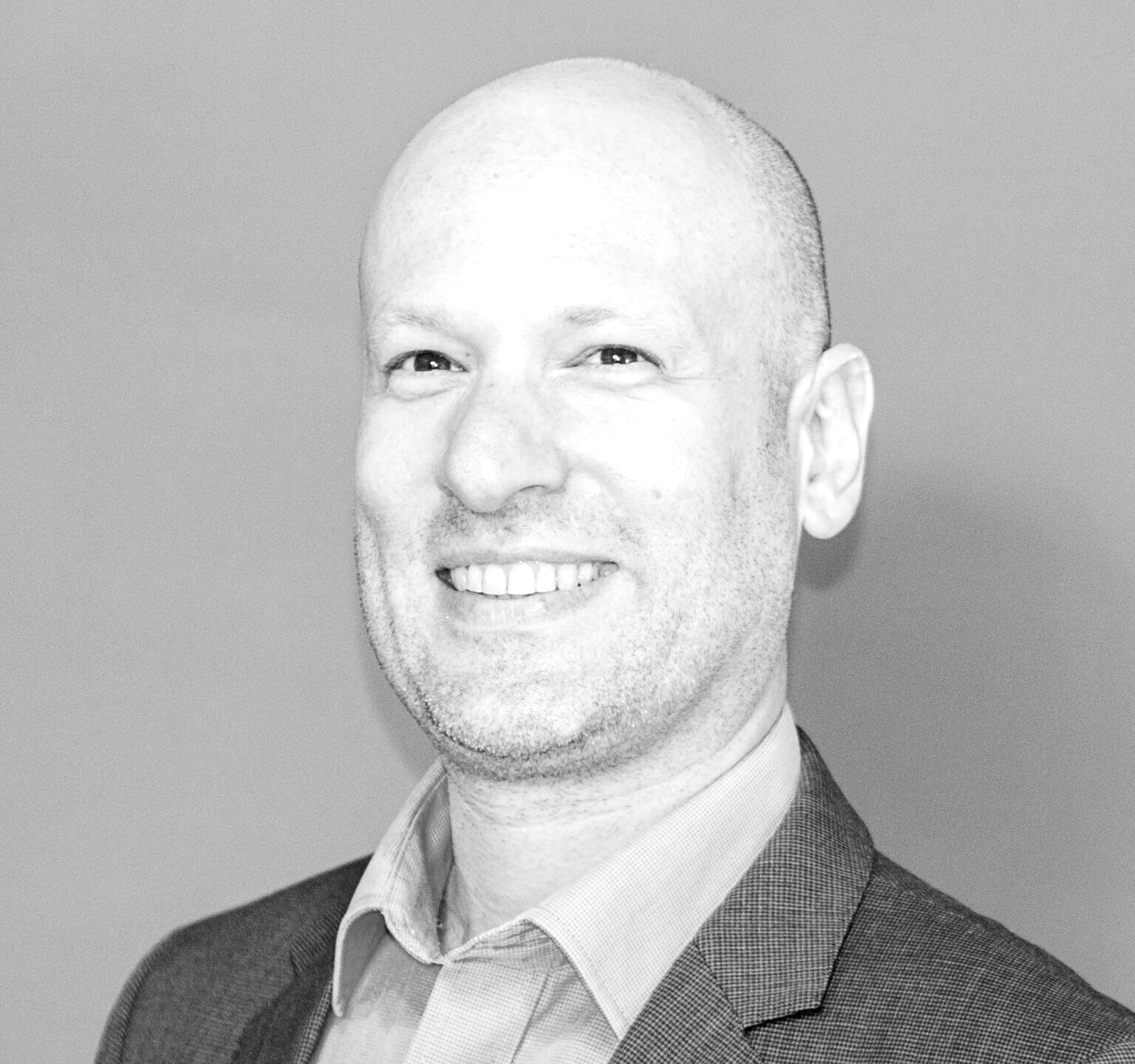}}]{Daniel Zelazo} (Senior Member, IEEE) received the B.Sc. and M.Eng. degrees in electrical engineering and computer science from the Massachusetts Institute of Technology, Cambridge, MA, USA, in 1999 and 2001, respectively, and the Ph.D. degree in aeronautics and astronautics from the University of Washington, Seattle, WA, USA, in 2009.\\
He is an Associate Professor of Aerospace Engineering and the Director of the Philadelphia Flight Control Laboratory, Technion--Israel Institute of Technology, Haifa, Israel. From 2010 to 2012, he was a Postdoctoral Research Associate and a Lecturer with the Institute for Systems Theory and Automatic Control, University of Stuttgart, Stuttgart, Germany. His research interests include topics related to multiagent systems.\\
Dr. Zelazo is currently an Associate Editor of \textsc{IEEE Control System Letters} and a Subject Editor of the International Journal of Robust and Nonlinear Control.
\end{IEEEbiography}

\end{document}

%% file: shortcuts_math.tex
\renewcommand{\emptyset}{\varnothing} 
\newcommand{\Rset}[1]{\mathbb{R}^{#1}} 

\def \c {\mathbf{c}}

\def \p {\mathbf{p}}

\def \x {\mathbf{x}}

\def \F {\mathbf{F}} 
 
\def \H {\mathbf{H}}

\def \L {\mathbf{L}}

\def\Amc{\mathcal{A}}

\def\Emc{\mathcal{E}}

\def\Gmc{\mathcal{G}}

\def\Imc{\mathcal{I}}

\def\Mmc{\mathcal{M}}
\def\Nmc{\mathcal{N}}

\def\Rmc{\mathcal{R}}

\def\Vmc{\mathcal{V}}
\def\Wmc{\mathcal{W}}

\newcommand{\Bdelta}{\boldsymbol{\delta}}

\newcommand{\Bomega}{\boldsymbol{\omega}}

\newcommand{\BDelta}{\boldsymbol{\Delta}}

\newtheorem{defn}{Definition}[section]

\newtheorem{thm}{Theorem}[section]
\newtheorem{lem}{Lemma}[section]
\newtheorem{prop}{Proposition}[section]
\newtheorem{rmk}{Remark}[section]
\newtheorem{fact}{Fact}[section]

\newtheorem{prb}{Problem}[section]
\newtheorem{cor}{Corollary}[section]

\newcommand{\normp}[2] {\left\|{#1}\right\|_{#2}}
\newcommand{\angb}[1] {\left<{#1}\right>}

\def \ones			{{\mathds{1}}} 
\def \zeros			{{\mathbf{0}}} 
\newcommand{\onesvec}[1] 	{\ones_{#1}} 
\newcommand{\zerovec}[1] 	{\mathbf{0}_{#1}} 

\DeclareMathOperator{\diag}{diag} 

\newcommand{\eye}[1]{I_{#1}} 
\renewcommand{\vec}[3]{\mathrm{vec}_{#1}^{#2}(#3)} 
\renewcommand{\diag}[3]{\mathrm{diag}_{#1}^{#2}(#3)} 

\def \kronecker {\otimes} 

\def \graph		{\mathcal{G}} 
\newcommand		{\neighbors}[1]{\mathcal{N}_{#1}}	